\documentclass[%
 reprint,
 amsmath,amssymb,
 aps,
]{revtex4-2}

\DeclareMathOperator{\negl}{negl}

\usepackage{amsthm}
\newtheorem{theorem}{Theorem}
\newtheorem{lemma}{Lemma}

\newtheorem{definition}{Definition}

\newtheorem*{Game*}{Game}
\newtheorem*{TokGen*}{TokenGeneration Phase}
\newtheorem*{TokVer*}{TokenVerification Phase}
\newtheorem*{Setup*}{Setup}
\newtheorem*{Query*}{Query}
\newtheorem*{Challenge*}{Challenge}
\newtheorem*{Guess*}{Guess}
\usepackage{breqn}
\theoremstyle{remark}

\usepackage{algorithm}
\usepackage{algpseudocode}

\usepackage{hyperref}
\hypersetup{
    hidelinks = true
}

\DeclareMathOperator{\sins}{sins}
\DeclareMathOperator{\sinc}{sinc}
\DeclareMathOperator{\E}{E}

\newcommand{\projector}[1]{\ketbra{#1}{#1}}

\newcommand{\identity}{\mathbbm{1}}

\usepackage{comment}

\usepackage{tikz}
\usetikzlibrary{quantikz}

\usepackage[capitalize]{cleveref}
\usepackage{tensor}
\usepackage{algpseudocode}
\usepackage{mathtools,physics, bbm}
\usepackage{graphicx}
\usepackage{dcolumn}
\usepackage{bm}

\begin{document}

\preprint{APS/123-QED}

\title{Existential Unforgeability in Quantum Authentication From Quantum Physical Unclonable Functions Based on Random von Neumann Measurement}

\author{Soham Ghosh$^{*}$, Vladlen Galetsky$^{*}$, Pol Julià Farré$^{\dagger}$, Christian Deppe$^{\dagger}$, Roberto Ferrara$^{*}$, Holger Boche$^{*}$}%
\affiliation{%
 $^{*}$Technical University of Munich \\
 $^{\dagger}$Technical University of Braunschweig
}%

\date{\today}

\begin{abstract}
Physical Unclonable Functions (PUFs) leverage inherent, non-clonable physical randomness to generate unique input-output pairs, serving as secure fingerprints for cryptographic protocols like authentication. Quantum PUFs (QPUFs) extend this concept by using quantum states as input-output pairs, offering advantages over classical PUFs, such as challenge reusability via public channels and eliminating the need for trusted parties due to the no-cloning theorem. Recent work introduced a generalized mathematical framework for QPUFs. It was shown that random unitary QPUFs cannot achieve existential unforgeability against Quantum Polynomial Time (QPT) adversaries. Security was possible only with additional uniform randomness. To avoid the cost of external randomness, we propose a novel measurement-based scheme. Here, the randomness naturally arises from quantum measurements. Additionally, we introduce a second model where the QPUF functions as a nonunitary quantum channel, which guarantees existential unforgeability. These are the first models in the literature to demonstrate a high level of provable security. Finally, we show that the Quantum Phase Estimation (QPE) protocol, applied to a Haar random unitary, serves as an approximate implementation of the second type of QPUF by approximating a von Neumann measurement on the unitary's eigenbasis. 
\end{abstract}

\maketitle

\section{Introduction}
\subsection{Physical Unclonable Function (PUF):}
   Physical Unclonable Functions (PUF) are hardware devices with the assumption of having randomness which is inherently both physical and unclonable. This results in each PUF having a unique behaviour that produces different outputs (responses) over the same set of inputs (challenges). The assumption is a theoretical abstraction of the idea that some unpredictable manufacturing variations can create devices for which replicating their behaviour is more easily done by simply querying (collecting challenges and responses) and learning about the device rather than trying to learn about the underlying manufacturing variations (sometimes destructively). This assumption is not only valid but also cost-effective, capitalizing on the inherent unpredictability in real-world manufacturing processes. In general, the manufacturing variations can be measured as a certain amount of randomness in bits or qubits that makes the amount of queries needed to learn about the PUF device exponential in the size of the input bit or qubit space, making the PUF device highly suitable for security schemes like authentication or encryption. The queries (often called \emph{Challenge-Response Pairs} or CRPs \cite{Arapinis2021quantumphysical}) can be stored as smaller unique fingerprints. 
    
    This assumption is not only valid but also cost-effective, capitalizing on the inherent unpredictability in real-world manufacturing processes. Manufacturing variations used to create PUFs encompass diverse examples, including gate delays within integrated circuits \cite{Gao2020,M.Roel} and imperfections within optical crystals \cite{Gao2020,Kim2022,M.Roel}. 

\subsubsection{Classical PUF (CPUF):}
  PUFs, recognized as an efficient security solution for establishing trust in communication systems \cite{N1_Holger}, have been extensively studied in classical contexts. Their role as a foundation for security in embedded systems is outlined in \cite{N4_Holger}, and a comprehensive overview of classical PUF constructions can be found in \cite{M.Roel}. Despite their widespread use, classical PUFs (CPUFs) face three significant limitations: 
    \begin{itemize}
        \item {Trusted Party:} Classical bits are clonable. Thus, the party which creates and/or stores the CRPs in any security scheme becomes a trusted party posing potential threats to privacy leakage.
        \item {Challenge reusability:} Again, as classical bits are clonable, no data point in the CRPs can be reused after communication via a public channel, as this would allow an eavesdropper to simply reuse the eavesdropped responses. 
        \item {Security:} Most of CPUFs constructions are heuristic, having no provable security. 
    \end{itemize}
As outlined in the following section, these limitations are effectively addressed through a quantum solution. Owing to these technological advantages, quantum technologies are central to advancing 6G communication systems, as discussed in \cite{N2_Holger,N3_Holger}.

\subsubsection{Quantum PUF (QPUF)}
    A Quantum PUF (QPUF) can simply be thought of as a device that produces  Quantum CRPs, having quantum states as data points.
    While the unclonability of the PUF protects the whole device in quantum and classical devices alike, in the CPUF past queries can be cloned.
    In contrast, the quantum no-cloning theorem ensures that even the single queries of the QPUF cannot be cloned.
    Moreover, performing quantum state tomography on a specific query would also require an exponential, in the size of the considered quantum system, number of identical queries.
    Thus the issue of \emph{trusted party} holding a challenge and \emph{challenge reusability} can be solved with QPUFs.
    
    A precursor of QPUFs~\cite{Skori2010} provided such security only against classical adversaries. Recent work~\cite{Arapinis2021quantumphysical} introduced a mathematical framework that defines QPUFs with provable security also against quantum adversaries.
    Since then, QPUFs have seen increased interest~\cite{DKDK21,PSATG21,GGDF22,DKKC22,CDMWAK23}. 
    
    In~\cite{Arapinis2021quantumphysical} it was proved that a QPUF can only be an approximation of a Haar (uniform) random unitary channel if and only if the domain and range of the QPUF channel are the same.
    They established three hierarchical security notions: \emph{Quantum Exponential Unforgeability, Quantum Existential Unforgeability} and \emph{Quantum Selective Unforgeability}, arranged in descending order of security strength.
    Subsequently, they explicitly proved  that their model inherently lacks both exponential and existential unforgeability.
    However, they proved that their model does exhibit selective unforgeability, a property that is often highly suitable for numerous practical applications.

\subsection{Our Contributions}
    Quantum exponential unforgeability is a security measure defined for any adversary possessing the power of querying and possessing exponentially many data points. However, this level of security can never be achieved \emph{realistically} by any classical or quantum PUF as, if an  exponential amount of queries were available, the adversary would be able to perform efficient state tomography on the PUF CRPs or process tomography on the PUF channel. 
    
    However, both Existential and Selective unforgeability are security measures against a Quantum Polynomial Time adversary, which is more plausible to be required. Since the authors in \cite{Arapinis2021quantumphysical} proved that no Haar random unitary channel can provide Quantum Existential Unforgeability, the only way to achieve that level of security would be to introduce non-unitarity in the QPUF construction.
    
    A simple way to construct such a non-unitary channel (quantum process) is to combine measurements with unitary evolutions.
    This approach is without loss of generality, as any quantum channel can be thought of as a measurement channel with a post measurement state, where some of the classical information from the measurement is lost.
    Thus, we modelled our QPUF using a general quantum instrument which performs a measurement on a Haar random basis \cite{M.Ozlos}. The QPUF performs a von Neumann measurement on an input quantum state but the measurement basis is rotated to some random basis state using a Haar random unitary \cite{M.Ozlos}. The measurement and the random rotation together constitute our QPUF and hence it is inherently a non-unitary quantum process.
    In short, we are proposing the idea that measurements in a QPUF can increase the security achievable in schemes using the aforementioned QPUF.
    Intuitively, the measurement helps because it collapses the quantum state of any adversary attempting to use the QPUF in a way that was not intended, while not affecting the quantum state of the intended users.

    We prove the  main result that our QPUF based on a Haar random von Neumann measurement possesses Quantum Existential Unforgeability. With this we introduce the first ever model for a QPUF that has such high level of provable security. 
    
    For the implementation of an ideal QPUF we propose two device models, each introducing its own limitation for a practical realization:
    \begin{itemize}
        \item A random measurement can be approximated by the phase estimation protocol \cite{Kitaev96} on a Haar unitary $U$.
        The problem with this model is that it assumes the device can perform $CU$ an exponential number of times with respect to the number of input qubits.
        \item A random measurement can be obtained by a standard basis measurement of the target state after it has been passed through a generalized controlled NOT gate with the control basis rotated by a Haar random unitary $U$ onto a random basis.
        This model assumes that the device can perform the inverse of a random unitary. Although there are known algorithms \cite{sardharwalla2016universal,Sedl_k_2019,Quintino_2019_Prob,Quintino_2019,Dong_2021,Quintino2022deterministic} for performing the inverse of a given unknown unitary, all of them have exponential gate cost complexity.
    \end{itemize}
    Such limitations would have to be addressed when designing practical realisations of QPUFs. We have discussed these limitations in details in \cref{Discussion} and proposed possible research directions in order to address them.

Another model proposition of ours follows by noticing that, in the case of Selective Unforgeability \cite{Arapinis2021quantumphysical}, the security comes from the fact that the verifier uses additional randomness apart from the QPUF to select their queries for the QPUF. This weakens the security notion, as achieving perfect uniform randomness is challenging in practice and incurs additional costs. Therefore, we propose to obtain this selection of query states via measurements which are natural sources of randomness according to the laws of quantum mechanics. Below, we describe a single round of the protocol. 
Namely, the verifier generates $M$ Bell states, $\ket{\Phi^{+}}$,
\begin{equation}
    \ket{\Phi^{+}} \equiv \frac{1}{\sqrt{D}}  \sum_{i \in \mathbb{Z}_{D}} \ket{i} \otimes \ket{i}
\end{equation}
and applies the QPUF unitary $U$ on one half of the state and leaves the other half unchanged. As a result, they obtain the states,
\begin{equation}
 \frac{1}{\sqrt{D}} \sum_{i \in \mathbb{Z}_{D}} \ket{i} \otimes U\ket{i} \equiv \ket{Q}
\end{equation}
The verifier stores this resulting set of $M$ states $\ket{Q}$ and returns the QPUF $U$ to the prover. During verification, the verifier measures the unchanged part of all the states to get the outcomes $\{m_{i}\}_{i \in \mathbb{Z}_{M}}$ and corresponding basis state $\{U\ket{m_{i}}\}_{i \in \mathbb{Z}_{M}}$ on the unmeasured halves. The verifier then sends the computational basis states $\ket{m_{i}}$ to the prover, who applies their QPUF to it and returns the response states. Finally, the verifier checks the match between the remaining half of the Bell states and the responses using a SWAP test \cite{SWAP_Mina}. If for all the $M$ SWAP test checks the test is passed, then the verifier accepts, otherwise rejects.
Security in this model relies on the fact that, before measurement, even the verifier does not know which basis state will be chosen. For the adversary to achieve an overall success probability exceeding order of \(\frac{1}{2^M}\) in the \(M\) SWAP Tests, they must guess states with high overlap with the correct basis states. However, the correct basis state corresponds to a 1-dimensional subspace of the Hilbert space, which has a measure of \(\frac{1}{\text{dimension of the Hilbert space}}\). For \(n\)-qubit systems, this dimension is \(2^n\), making the probability of the adversary guessing correctly vanishingly small.

Using the additional entanglement from Bell states is efficient and not more costly than previous approaches for the following reasons:
\begin{itemize}
    \item The verifier stores both halves of the Bell state locally, making it easier to maintain entanglement in practical applications.
    \item In previous models, states produced by applying Haar random unitaries are stored. Since Haar random unitaries naturally create highly entangled states, the process already involves managing such states. Therefore, our proposal does not add any extra cost in comparison.
\end{itemize}
    \subsection{Future Work}

Being these new QPUF models, there are several directions for future work that needs to be investigated:
\begin{itemize}
    \item As mentioned earlier, the methods of implementing a random measurement through a random measurement channel is problematic for a realistic implementation.  Finding more realistic implementation of the random measurement QPUF is thus an important open question. 
    \item Creating a manufacturing process that produces Haar or close-to-Haar unclonable unitaries might be impossible. We considered the Haar-unitary model as just the simplest theoretical proof of concept model showing that unforgeability can be achieved. For practical applications, more realistic unitary distributions still allowing for a theoretical security proof need to be found. 
    However, if the adversary is restricted to only query access, \textit{Pseudo Random Unitaries} (PRUs) \cite{doosti2022connection} or \textit{Unitary Designs} \cite{boche2019simultaneous, boche2017randomness} can serve as QPUFs, as demonstrated in \cite{doosti2022connection}. This is because the query database generated by PRUs or unitary designs is nearly indistinguishable from that of Haar random unitaries. However, in practice, they fall short as a solution, since PRUs can be implemented using efficient-depth quantum circuits, enabling an adversary with physical device access to learn them gate by gate.
\item We have considered only noiseless PUFs, thus cost of dealing with noise in this QPUF still needs to be
investigated.
\end{itemize}
    Finally, we highlight that QPUFs are advantageous over competing schemes like $\lq$\emph{Quantum Key Distribution} (QKD)'\cite{Metger2023, upadhyaya2021dimension} and $\lq$\emph{Quantum Money}' \cite{gavinsky2012quantum, Wiesner}.  In contrast to our QPUF model, a fundamental requirement in Quantum Money and QKD based schemes is that the classical information about the measurement basis needs to be stored \textbf{secretly} by the verifier. This makes them a $\lq$\emph{trusted party}' in the security scheme, opening avenues for possible privacy leakage.  

\section{Notation and Preliminaries}
We denote random variables with uppercase letters (e.g. $X$) and for pure quantum states we follow the Dirac bra-ket (e.g. $\ket{\psi}$) notation. Lower case Greek letters (e.g. $\rho$) are used for labelling density operators.

The trace distance \cite{NielsenChuang2010},\cite{M.Wilde2016}  between two density matrices $\rho$ and $\sigma$ is denoted as $0 \leq D(\rho,\sigma) \leq 1$. The uppercase Greek letter $\Lambda$ is used to denote a completely positive trace preserving (CPTP) map or quantum channel \cite{NielsenChuang2010}, \cite{M.Wilde2016}. The diamond norm distance \cite{M.Wilde2016} between two such quantum channels $\Lambda_{1}$ and $\Lambda_{2}$ is denoted as $||\Lambda_{1} - \Lambda_{2}||_{\diamond}$ . 

$\negl(\lambda)$ denotes a negligible function of lambda that decays exponentially to $0$ with increasing values of $\lambda$, e.g., $\frac{1}{2^{\lambda}}$ is a negligible function of $\lambda$. 

For any integer $D$, we define the set $\mathbb{Z}_{D}$ as follows:
\begin{equation}
    \mathbb{Z}_{D} \coloneqq \left[0, \cdots , D-1 \right].
\end{equation}

\subsection{Computational Security}

In contrast to the information-theoretic concept of $\lq$perfect secrecy' \cite{book}, $\lq$computational security' \cite{book} considers the limitations on the computational power of the adversary and also permits a negligible probability of success for the adversary in breaching the security scheme. Perfect secrecy schemes do not put any restrictions on the computational power of an adversary and also do not allow for any non-zero probability of success for the adversary in breaching the scheme, even with \emph{unlimited computational power}! These two relaxations make computational security schemes practical and feasible. Hence, in our analysis we would also opt for this approach. 


\subsubsection{Computational Security Parameter}

Any cryptographic scheme can be conceptualized as a sequence of problems, and breaching the scheme entails solving this sequence. The $\lq$\emph{computational security parameter}'
\cite{book} is an integer valued parameter, denoted by the Greek letter $\lq$$\lambda$' and often visualised as a bitstring of ones $\textbf{1}^{\lambda}$, which serves to parameterize the spacial and/or time complexity  of the problems within a cryptographic scheme. In most applications, such as key or token-based schemes, the security parameter can be perceived as the length of the input bitstrings to the key or token generation protocol. This parameter also characterizes the computational capabilities of both the trusted party and the adversary.

For a cryptographic scheme to achieve computational security, a fundamental requirement is that the complexity of the problems underlying the scheme should \emph{not} belong to the \emph{polynomial class} \cite{book}, while the computational power of both the trusted party and the adversary should fall within the \textit{polynomial class} \cite{book}. The general proof strategy for establishing computational security involves calculating or bounding the probability of  success for the adversary in breaching the scheme as a function of the security parameter $\lambda$ and demonstrating an exponential decay of this probability with the scaling of the parameter. This shows the practicality of this approach by providing a clear understanding of how security scales in relation to the scaling of resources.

\subsection{Haar Random von Neumann Measurement Channel}

For a $D$-dimensional Hilbert space $\mathbb{C}^{D}$, the computational basis $\{\ket{i}\}_{i = 0}^{D-1}$ forms a complete basis set. It follows trivially that for any Haar random unitary $U$ \cite{M.Ozlos}, we can obtain a complete basis set $\{U\ket{i}\}_{i = 0}^{D-1}$. The measurement channel for measurement on such a basis is obtained from the quantum instrument \cite{M.Wilde2016},
\begin{equation}
\label{eq: Channel_Def}
    \Lambda_{U}(\rho) = \sum_{i \in \mathbb{Z}_{D}} \ketbra{i}{i} \otimes \Pi_{i}^{U} \rho \Pi_{i}^{U \dagger},
\end{equation}
where the projection operators $\Pi_{i}^{U}$ are defined as
\begin{equation}
\label{eq:PVM}
    \Pi_{i}^{U} \equiv U \projector{i}U^{\dagger}.
\end{equation}
When tracing the first system, the quantum instrument becomes a measurement channel described by the Projection Valued Measures (PVMs) \cite{M.Wilde2016} $\{\Pi_{i}^{U}\}$. 
The post-measurement state after getting a measurement outcome $i$ becomes
\begin{equation}
\label{Haar token}
    \Pi_{i}^{U} \rho \Pi_{i}^{U\dagger} = U\projector{i}U^{\dagger},
\end{equation}
where
\begin{equation}
\label{eq: Token_prob}
    p_{i} = \Tr \left[\Pi_{i}^{U} \rho \right]
\end{equation}
is the outcome probability for a measurement outcome $i$.


\subsection{Quantum Phase Estimation as an Approximate von Neumann Measurement}
\label{QPE}
We highlight the preliminaries of Quantum Phase Estimation Algorithm \cite{Kitaev96} to show that it can be used as a model to approximate a Haar random von Neumann measurement. As a result, in \cref{PE-QPUF} we will explain how it can be used to derive a quantum circuit based implementation for our proposed QPUF.   
\subsubsection{Phase Measurement}

 Let $U$ be any unitary.
    While $U$ is not Hermitian, it is still normal, meaning that $U$ is unitarily diagonalizable (equivalently, its Hermitian and anti-Hermitian parts commute) and, thus, it has a spectral decomposition and defines a quantum measurement. 
    The eigenvalues of a unitary are complex roots of unity (elements of the complex unit circle).
    Equivalently, we can define the measurement of the Hermitian operator $H$ such that $U=e^{iH}$, namely the measurement of $-i\log U$.
    This changes the measurement outcomes (which can be done with classical processing, since the outcomes are classical values), but preserves the measurement projectors.
    For convenience, we choose a dimension $d$ and we consider the measurement of the hermitian operator
    \[\phi \coloneqq \frac{\log U}{i\frac{2\pi}{d}},\]
    with eigenvalues fulfilling $\phi_i \in [0,d[$ when we write all the eigenvalues of $U$ in the form $e^{i\frac{2\pi}{d} \phi}$.
    Below, $d$ will become both the dimension of an ancilla system and the amount of precision we use to estimate each of the complex phases of the eigenvalues of $U$.

\subsubsection{Phase Estimation Algorithm}
\label{sec:phase-estimation-algorithm}
    Let $d$ be the dimension of an additional quantum system $\mathbb{C}^{d}$, ancillary to the system of $U$, and let $F$ and $F^\dagger$ be the Fourier Transform and the inverse Fourier Transform respectively in this $d$-dimensional system.
    The phase-estimation~\cite{Kitaev96} circuit for $U$ with outcomes in $\mathbb{Z}_{d}$ is:
    \begin{equation}
    \begin{quantikz}
    &&
    \lstick{$\ket{0}$} &\gate{F^\dagger}&\ctrl{2}  &  \gate{F} & \meterD{\ket{k}} & \cw & k  
    \\
    U_k =&&
    &           &             &           &             &    &   &   &
    \\
    &&
    &\qw &\gate{U}& \qw & \qw
    \end{quantikz}
    \label{eq:kraus-circuit}
    \end{equation}
    where $CU$ is the controlled unitary defined as
    \begin{equation}
        CU = \sum_{k\in\mathbb{Z}_{d}} \projector{k} \otimes U^k. \notag \label{eq: CU}
    \end{equation}
    Note that we are making an abuse of notation here, as we are not performing a standard controlled unitary. This circuit also corresponds to the $d$ Kraus operators $U_k$ of the quantum instrument
    \begin{equation}
    \label{eq: QPE Channel}
    \Lambda_{U}^{QPE}(\rho) = \sum_{k\in\mathbb{Z}_{d}} \projector{k} \otimes U_k \rho U_k^\dagger    
    \end{equation}
    produced by the circuit. 
    
    Notice, that repetitions of this algorithm always commute, independently of $d$ or $U$.
    Indeed, all Kraus operators and, thus, also any two instances of applying $\Lambda_{U}^{QPE}$, commute with each other because they overlap only on the target where the various powers of $U$ commute, as shown in the circuit below
    \begin{equation*}
    \begin{quantikz}
    \lstick{$\ket{0}$} &\gate{F^\dagger}&\ctrl{2}  &  \gate{F} & \meterD{\ket{k}} & \cw \ k  
    \\
    &
    &           &             &           &             &    &   &   &
    \\
    U_k U_l = 
    & &\gate{U}& \gate{U} & \qw &
    = U_l U_k.
    \\
    &
    &           &             &           &             &    &   &   &
    \\
    &
    \lstick{$\ket{0}$} &\gate{F^\dagger}&\ctrl{-2}  &  \gate{F} & \meterD{\ket{l}} & \cw \ l
    \end{quantikz}
    \end{equation*}
    This property will simplify later the analysis of the QPUF in \cref{PE-QPUF}.
    Notice, also, that the same argument gives that the $U_k$'s commute with $U$ itself, meaning $U_k$ are also diagonal in the same basis.
    A consequence of this is that the conditional un-normalized states
    \[U_{k_n} \dots U_{k_{0}} \ket\psi\]
    and the probabilities
    \begin{align*}
      p_{k_{0} \dots k_n}
    &= \bra\psi U_{k_{0}}^\dagger \dots U_{k_n}^\dagger U_{k_n} \dots U_{k_{0}} \ket\psi \\
    &= \bra\psi \left(U_{k_{n}}^{\dagger}U_{k_{n}}\right) \dots \left(U_{k_{0}}^{\dagger}U_{k_{0}}\right) \ket\psi  
    \end{align*}
    are all symmetric in $k_{0}, \dots, k_n$.


\subsubsection{Phase Estimation Operators}

 We write the unitary $U$ and the Fourier transform $F$ as 
    \begin{align}
        U &= \sum_{i\in \mathbb{Z}_{D}} e^{i \frac{2\pi}{d}\phi_{i}} \projector{\phi_{i}}
        &
        F &= \frac{1}{\sqrt{d}}\sum_{x,y \in \mathbb{Z}_{d}} e^{- i \frac{2\pi}{d} xy} \ketbra{x}{y},
    \end{align}
    where $d$ is the dimension of the ancillary control system  which $F$ acts on, $D$ is the dimension of the target system which $U$ acts on, and $\ket{\phi_{i}}$ are the eigenstates of U, with corresponding eigenvalue $e^{i \frac{2\pi}{d}\phi_{i}}$, where each $\phi_{i} \in [0,d[$  is \emph{unknown}.
    Notice how we use $d$ and not $D$ to define each $\phi_{i} $, which is due to the fact that the precision of the phase estimation is defined by the ancillary control system and not the original target system.
    \\
    The explicit form of the Kraus operators in the circuit of \cref{eq:kraus-circuit} is
    \begin{align}
        U_k &
        = \frac{1}{d} (\bra{k}F\otimes\identity) CU \qty(F^\dagger\ket{0} \otimes \identity)
        \\&
        = \frac{1}{d} \sum_{x\in \mathbb{Z}_{d}} \bra{k}F\ket{x}     \cdot U^x
        = \frac{1}{d} \sum_{x\in \mathbb{Z}_{d}} e^{-i\frac{2\pi}{d}kx} \cdot U^x
        \\&
        = \sum_{j\in \mathbb{Z}_{d}} 
            \qty( \frac{1}{d} \sum_{x\in \mathbb{Z}_{d}} e^{i\frac{2\pi}{d}(\phi_{j}-k)x} ) 
            \projector{\phi_{j}}.
    \end{align}
    When $\phi_{j}$ is an integer, the geometric series in parenthesis is the Kronecker delta $\delta_{k,\phi_{j}}$.
    In all other cases (in general whenever $\phi_{j}-k\neq 0$), we have 
    \begin{align}
        \sum_{x \in \mathbb{Z}_{d}} e^{i \frac{2\pi}{d} (\phi_{j}-k)x}
        &=
        \frac{1-e^{i2\pi(\phi_{j}-k)}}{1-e^{i \frac{2\pi}{d} (\phi_{j}-k)}} \\
        &=
        \frac{e^{i\pi(\phi_{j} -k)}}{e^{i\frac{\pi}{d}(\phi_{j}-k)}}
        \frac{\sin(\pi(\phi_{j}-k))}{\sin(\frac{\pi}{d} (\phi_{j}-k))}. \notag
    \end{align}
    When divided by $d$, this term is an approximation of the normalized $\sinc$ function, if $d$ is sufficiently large.
    
    Namely, we can define (with a rescaling of the input by $\pi$ for convenience)
    \begin{align}
        \sins_d(x) &\coloneqq
        \begin{cases}
            1 & x=0
            \\
            \frac{\sin(\pi x)}{d\sin(\frac{\pi}{d} x)}& x \neq 0
        \end{cases}
        &
        \sinc(x) &\coloneqq
        \begin{cases}
            1 & x=0
            \\
            \frac{\sin(\pi x)}{\pi x}& x \neq 0,
        \end{cases}
    \end{align}
  and have the following bound, coming from $\sin x \leq x$:
    \[\sins_d(x) \geq \sinc (x).\]
    Notice that both functions are zero when $x$ is a non-zero integer and that $\sins_d$ is symmetric just like $\sinc$.
    We can finally write the Kraus operators as
    \begin{align}
    \label{eq: QPE Kraus Operators}
        U_k &
        = \sum_{j\in \mathbb{Z}_{D}}
            \frac{e^{i\pi(\phi_{j}-k)}}{e^{i\frac{\pi}{d}(\phi_{j}-k)}}
            \sins_d(\phi_{j}-k)
            \projector{\phi_{j}},
    \end{align}
    which \emph{are diagonal in the eigenbasis of $U$}, as we anticipated in \cref{sec:phase-estimation-algorithm}.
    When tracing the first quantum system, the quantum instrument $\Lambda_{U}$ becomes a measurement channel with POVM \cite{M.Wilde2016} elements 
    \begin{align}
        M_k 
        \coloneqq U_k^\dagger U_k
       & = \sum_{j\in \mathbb{Z}_{D}}
            \sins^{2}_{d}(\phi_{j}-k)
            \projector{\phi_{j}}.
    \end{align}
Notice that, while $\sins_d$ has period $2d$, $\sins^{2}_d$ has period $d$, which is why the sum in $\sum_k M_k = \identity$ can be taken in any contiguous set of $d$ integers, namely in $N, N+1, ..., N + d -1$ for any integer $N$. For completeness, we note that the isometry $V_U$ before the measurement is given by
    \begin{align}
        V_U &
        = 
            \begin{quantikz}
            \lstick{$\ket{0}$} &\gate{F^\dagger}&\ctrl{1}  &  \gate{F} & \qw
            \\
            & &\gate{U}& \qw
            \end{quantikz}
        = \sum_{k \in \mathbb{Z}_{d}} \ket{k} \otimes U_k
        \\&
        = \sum_{k \in \mathbb{Z}_{d},\ j\in \mathbb{Z}_{D}}
            \frac{e^{i\pi(\phi_{j}-k)}}{e^{i\frac{\pi}{d}(\phi_{j}-k)}}
            \sins_d(\phi_{j}-k)
            \ket{k} \otimes \projector{\phi_{j}}.
    \end{align}
    
\subsubsection{Phase Estimation Outcome}

In the special case when each $\phi_i$ is an  integer value, namely when the eigenvalues of $U$ are integer powers of $e^{i \frac{2\pi}{d}}$, then the phase estimation exactly produces the measurement defined by $U$ with preservation of the post measurement state. Thus, it becomes equal to a von Neumann measurement on the eigenbasis of $U$.
    For all other cases, the outcome will be a probabilistic and integer estimation of the quantity $\phi_i$.
    
    It can be now intuitively understood why repeated phase estimations commute.
    Indeed, when $\Lambda_{U}^{QPE}$ is exactly the measurement defined by $U$ then it is clear that different applications of $\Lambda_{U}^{QPE}$ must commute, since they would be  equivalent to different measurements of the same observable.
    The commutativity statement in the previous section says that commutativity holds in general even for approximate measurements.
    
    Notice that phase estimation is a fully general model of measurements, because for any observable (Hermitian) $H$ we can create the following unitary
    \[e^{i \frac{H}{\norm{H}_\infty}},\]
    which has the exact same spectral projectors as $H$ (normalization by $\norm{H}_\infty$ is to guarantee that no phases loop around $\pm \pi$. Any constant above $\norm{H}_\infty / \pi$ would work).

\section{Results}
In this section, we describe our results.

\subsection{Measurement based QPUF (MB-QPUF) scheme}

As shown in \cref{fig:MB-QPUF}, the protocol can be described in the framework of standard QPUF/PUF based protocols \cite{Arapinis2021quantumphysical}, namely, it has two phases, a query phase done by the verifier and a verification phase.

\begin{figure*}
    \centering
    \includegraphics[scale = 0.8]{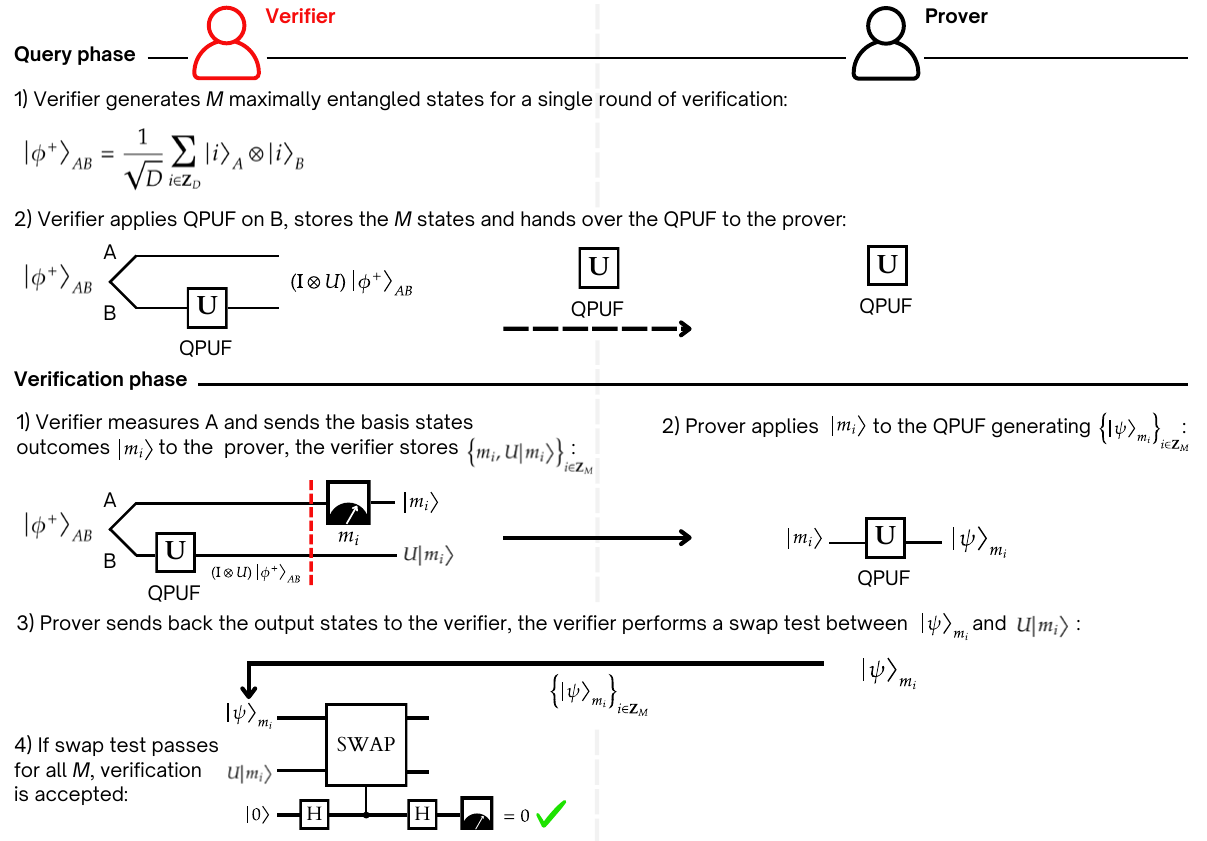}
    \caption{Single Round verification with $M$ trials. The query phases describes how the verifier stores $M$ Choi states of the QPUF $U$ and transfers the device to the prover. The verification phase describes, the challenge generation via measurements and subsequent verification with SWAP test. }
    \label{fig:MB-QPUF}
\end{figure*}

\subsubsection{Verifier's Query Phase}

The data collection of the verifier is described in the following steps:

\begin{itemize}
    \item Given $N,M = poly(\lambda)$(\textit{lowest degree of the polynomial is at least $1$}), the verifier initialises $NM$ many maximally entangled states on two $D-$dimensional systems labelled $A,B$, $\ket{\Phi^{+}}_{AB}$, where,
    \begin{equation}
    \ket{\Phi^{+}}_{AB} \equiv \frac{1}{\sqrt{D}}\sum_{i \in \mathbb{Z}_{D}} \ket{i}_{A} \otimes \ket{i}_{B}. 
    \end{equation}
    \item Then they apply the QPUF $U$ locally on system $B$ of all the maximally entangled states to obtain the following states,
    \begin{equation}
    \left(\mathbb{I} \otimes U \right)\ket{\Phi^{+}}_{AB} = \frac{1}{\sqrt{D}}\sum_{i \in \mathbb{Z}_{D}} \ket{i}_{A} \otimes U \ket{i}_{B} 
    \end{equation}
    \item Finally they store these states and handover the QPUF $U$ to the prover.
\end{itemize}

\subsubsection{Prover's Verification}

The setup describe here allows for $N$ rounds of verification where $M$ many maximally entangled states are used for each round. The detailed steps are described as follows:

\begin{itemize}
    \item When the prover wishes to verify their identity, the verifier measures the system label $A$ for $M$ maximally entangled states, obtaining $M$ measurement outcomes $\{m_{i}\}_{i \in \mathbb{Z}_{M}}$. This measurement will make the states in the other systems $\{U\ket{m_{i}}\}_{i \in \mathbb{Z}_{M}}$. 
    \item The verifier then sends the computational basis states corresponding to the obtained measurement outcomes $m_{i}$, i.e., $\ket{m_{i}}$ to the prover. 
    \item The prover then simply receives these states and applies their QPUF on these states and sends back response states $\{\ket{\psi}_{m_{i}}\}_{i \in \mathbb{Z}_{M}}$ to the verifier.
    \item The verifier then performs the SWAP Test \cite{SWAP_Mina} on all $M$ pairs, $U\ket{m_{i}}$ and $\ket{\psi}_{m_{i}}$. If the Swap test is passed for all states, the verification is accepted or else it is rejected. 
\end{itemize}
\subsubsection{Security}
Let $\mathcal{A}$ be any arbitrary adversary with query access to the QPUF $U$. Then they would be able to create a query database,
\begin{equation}
    Q_{\mathcal{A}} \coloneqq \{\rho^{i}_{in}, \rho^{i}_{out} \equiv U\rho^{i}_{in}U^{\dagger}\}_{i \in \mathbb{Z}_{|Q_{\mathcal{A}}|}}.
\end{equation}
Given a security parameter $\lambda$, we consider $|Q_{\mathcal{A}}| = poly(\lambda)$. In other words, we consider a Quantum Polynomial Time Adversary (QPT). For $|Q_{\mathcal{A}}| = \exp(\lambda)$, it is straightforward to see that no QPUF will be secure as the adversary would be able to perform efficient tomography.

Let the $M$ many measurement outcomes at the verifier's end be $\{m_{i}\}_{i \in \mathbb{Z}_{M}}$. Now, the winning strategy of the Adversary for passing all the $M$ rounds of the SWAP test could in general be an entangled state over $M$ systems. However, in \cite{SWAP_Mina} it was proven that the best strategy for the adversary would be to send separable states. Therefore, let the adversary prepare $M$ many guess states $\{\rho^{i}_{\mathcal{A}}\}_{i \in \mathbb{Z}_{M}}$. We can then write the overall winning probability of the adversary for a single round of verification,
\begin{align}
\label{eq: MB-QPUF prob}
    P_{\mathcal{A}} &= \prod_{i \in \mathbb{Z}_{M}} \left(\frac{1}{2} + \frac{1}{2}F^{2}(\rho^{i}_{\mathcal{A}},U\ketbra{m_{i}}U^{\dagger}) \right) \notag \\
    &= \prod_{i \in \mathbb{Z}_{M}} \left(\frac{1}{2} + \frac{1}{2}\left| \Tr[\rho^{i}_{\mathcal{A}}U\ketbra{m_{i}}U^{\dagger}] \right|^{2} \right) 
\end{align}
where $F(\cdot,\cdot)$ denotes the Uhlmann Fidelity \cite{NielsenChuang2010} between two states.
Notice that, since $U$ is a random variable distributed according to the Haar measure, $P_{\mathcal{A}}$ is also a random variable. Since, each round of verification is independent the overall success probability of the adversary will be a product of all the single round probabilities. For simplicity, we focus on a single round of verification, showing that its maximum success probability is negligible in the security parameter, and thus the overall success probability is even smaller.

\begin{definition}[Measurement Selective Unforgeability]
  A unitary QPUF, when used in the MB-QPUF model described above, is defined as \textit{measurement selective unforgeable} if,
  for any QPT adversary with query access to the QPUF, the overall success probability \( P_{\mathcal{A}} \) for a single round of verification is at most negligible in the security parameter.
\end{definition}
We then formally prove that any MB-QPUF is \textit{measurement selective unforgeable}.

\begin{theorem}[Measurement Selective Unforgeability]
For any security parameter \(\lambda\) and the number of trials \(M = poly(\lambda)\), in a single verification round, the expected success probability of any adversary is bounded by:
\[
\mathbb{E}[P_{\mathcal{A}}] \leq \frac{1}{2^M} + \negl(\lambda) = \negl(\lambda).
\]
Consequently, any MB-QPUF scheme is measurement-selective unforgeable.
\end{theorem}
\begin{proof}
From \cref{eq: MB-QPUF prob}, we have the probability of success as,
\begin{align*}
    P_{\mathcal{A}} &= \prod_{i \in \mathbb{Z}_{M}} \left(\frac{1}{2} + \frac{1}{2}\left| \Tr[\rho^{i}_{\mathcal{A}}U\ketbra{m_{i}}U^{\dagger}] \right|^{2} \right), \\
    &\leq \frac{\bigg(1 + \max_{i} \Tr[\rho^{i}_{\mathcal{A}}U\ketbra{m_{i}}U^{\dagger}]\bigg)^{M}}{2^{M}}, \\
    &\leq \frac{1 + \max_{i} \Tr[\rho^{i}_{\mathcal{A}}U\ketbra{m_{i}}U^{\dagger}] \left(\sum_{k = 1}^{M} \binom{M}{k}\right)}{2^{M}}, \\
    &= \frac{1 - \max_{i} \Tr[\rho^{i}_{\mathcal{A}}U\ketbra{m_{i}}U^{\dagger}]}{2^{M}}, \\ 
    &+ \max_{i} \Tr[\rho^{i}_{\mathcal{A}}U\ketbra{m_{i}}U^{\dagger}].
\end{align*}
Define, \begin{equation*}
    Q \equiv \max_{i} \Tr[\rho^{i}_{\mathcal{A}}U\ketbra{m_{i}}U^{\dagger}].
\end{equation*}
Therefore, we have, from the linearity of expectation,
\begin{align}
    \E[P_{\mathcal{A}}] &\leq \frac{1 - \E[Q]}{2^{M}} + \E[Q], \notag \\
    &\leq \frac{1}{2^M} + \E[Q].
\end{align}
Let \( \Pi_{\mathcal{A}} \) denote the projector onto the subspace spanned by states in the adversary's query database \( Q_{\mathcal{A}} \). The maximum dimension of this subspace is \( |Q_{\mathcal{A}}| \). For any state \( U\ket{m_i} \), the variable \( Q \) represents the probability with which the adversary successfully guesses it. 

Define \( \rho_{\mathcal{A}} \) as the state generated by the adversary if \( U\ket{m_i} \in \operatorname{span}(\Pi_{\mathcal{A}}) \), and \( \rho^{\perp}_{\mathcal{A}} \) as the state generated if \( U\ket{m_i} \in \operatorname{span}(\mathbb{I} - \Pi_{\mathcal{A}}) \). The variable \( Q \) can then be expressed as:
\begin{align}
    Q &= \Tr[\rho_{\mathcal{A}} U \ketbra{m_i} U^{\dagger}] \cdot \Tr[\Pi_{\mathcal{A}} U \ketbra{m_i} U^{\dagger}] \notag \\
    &\quad + \Tr[\rho^{\perp}_{\mathcal{A}} U \ketbra{m_i} U^{\dagger}] \cdot \Tr[(\mathbb{I} - \Pi_{\mathcal{A}}) U \ketbra{m_i} U^{\dagger}], \notag \\ 
    &\leq 1 \cdot Q_1 + Q_2 \cdot 1 .
\end{align}
where we define,
\begin{subequations}
    \begin{align}
        Q_1 &\equiv \Tr[\Pi_{\mathcal{A}} U \ketbra{m_i} U^{\dagger}], \\
        Q_2 &\equiv \Tr[\rho^{\perp}_{\mathcal{A}} U \ketbra{m_i} U^{\dagger}]. 
    \end{align}
\end{subequations}
Now, by linearity of expectation we have,
\begin{align*}
    \E[Q] &= \E[Q_1]+ \E[Q_2], \\
    &\leq \frac{|Q_{\mathcal{A}}|}{D} + \frac{1}{D - |Q_{\mathcal{A}}|} .
\end{align*}
where in the last step, we have applied \cref{Thm: Existential Unforgeability} to obtain the bound on $\E[Q_2]$.
Since, $D = \exp(\lambda)$, for example for n qubits, it is $2^{n}$, and $|Q_{\mathcal{A}}|,M = poly(\lambda)$ by definition, we have,
\begin{equation}
    \E[Q] \leq \negl(\lambda).
\end{equation}
Finally, we have,
\begin{align}
    \E[P_{\mathcal{A}}] &\leq \frac{1}{2^{M}} + \negl(\lambda) ,\notag \\
    &= \negl(\lambda).
\end{align}
This, completes the proof.
\end{proof}

\subsection{Ideal QPUF}

In this section we define a new model for QPUFs, namely we construct a non-unitary quantum channel which would serve as a QPUF.

\subsubsection{Hardware Requirements}

In accordance with the framework in \cite{Arapinis2021quantumphysical}, we redefine the essential hardware requirements for the Haar random von Neumann measurement channel to be considered an "\textit{ideal QPUF}".

\begin{definition}[Ideal QPUF]
\label{QPUF Theorem}
Any quantum instrument $\Lambda_{U}$ in \cref{eq: Channel_Def}, becomes an $\lq$\emph{ideal QPUF}' if it satisfies the following hardware requirements for some security parameter $\lambda$:
\begin{itemize}
 \label{Hardware_Req}
    \item
    \textbf{Unknownness:} The Kraus operators $\{U \ketbra{i}{i} U^{\dagger}\}_{i \in \mathbb{Z}_{D}}$ for the channel $\Lambda_{U}$ remain unknown.
    \item \textbf{Collision-Resistance:}
The probability of having any two post-measurement pure states $\ket{\psi_{i}}$ and $\ket{\psi_{j}}$ corresponding to different measurement outcomes $i \neq j$ being furthest apart in trace distance satisfies
\begin{equation}
\Pr\big[D(\ketbra{\psi_{i}}{\psi_{i}},\ketbra{\psi_{j}}{\psi_{j}}) = 1 \big] \geq 1 - \negl(\lambda)   
\end{equation} 
\item \textbf{Robustness:} 

The probability of having any two post-measurement pure states $\ket{\psi_{i}}$ and $\ket{\psi_{j}}$ corresponding to the same measurement outcomes $i = j$ being equal satisfies
\begin{equation}
\Pr\big[D(\ketbra{\psi_{i}}{\psi_{i}},\ketbra{\psi_{j}}{\psi_{j}}) = 0 \big] \geq 1 - \negl(\lambda).
\end{equation} 
\item \textbf{Uniqueness:} 

The probability of any two distinct QPUF instruments, $\Lambda_{U_{i}}$ and $\Lambda_{U_{j}}$, being furthest apart in diamond norm
\begin{equation}
\label{eq: Uniqueness}
\Pr \left[||\Lambda_{U_{i}} - \Lambda_{U_{j}}||_{\diamond} = 2 \right] \geq 1 - \negl (\lambda),
\end{equation}
where $2$ is the maximum value in diamond norm distance.
\end{itemize}
\end{definition}


\subsection{Query structure for ideal QPUF}
For any collection of input pure states $\mathcal{D} \coloneqq \{\ket{\psi_{i}}\}$,we define the query set of the \emph{ideal QPUF} $\Lambda_{U}$ as
\begin{equation}
\label{eq: Query Structure}
    Q \coloneqq \{\ket{\psi_{i}}, q_{i}, U \ket{q_{i}}\} ,
\end{equation}
where $q_{i} \in \mathbb{Z}_{D}$ is the measurement outcome corresponsding to the $i^{th}$ input state $\ket{\psi_{i}}$ and $U\ket{q_{i}}$ is the post-measurement state obtained after measurement through the instrument $\Lambda_{U}$. We denote the computational basis by $\ket{q_{i}}$ for the classical label $q_{i}$. As mentioned earlier \cref{eq: Channel_Def}, measuring with $\Lambda_{U}$ yields the basis $U\ket{q_{i}}$ for the classical outcome $q_{i}$, instead of $\ket{q_{i}}$. It is evident that $|Q| = |\mathcal{D}|$. This set $Q$ being made of the CRPs of \emph{ideal} QPUF $\Lambda_{U}$ provides a unique fingerprint. 

\subsection{Ideal Authentication Protocol}

\subsubsection{Token Generation Protocol}
\label{HelperConst_Ideal}
In $Q$, each state $U\ket{q_{i}}$ serves as a secure quantum token  with the classical value $q_{i}$ as its identifier, which is kept public as shown in \cref{Fig:PUF_construct} (red). Each token in $Q$ is sent to different authentic parties, who store them in their quantum memory. Notice that, since $U\ket{q_{i}}$ is a quantum state and $q_{i}$ is public, there is no \emph{trusted party} in the protocol.


\subsubsection{Token Verification Protocol}
Upon verification, the token state $U\ket{q_{i}}$ is requested back by the verifier. The state is then used as input state for the \emph{ideal QPUF} channel, $\Lambda_{U}$ in \cref{QPUF Theorem}, to obtain the pair of measurement outcome and post-measurement state as $(q'_{i}, U\ket{q'_{i}})$ after measurement as shown in \cref{Fig:PUF_construct} (blue). If $q_{i} = q'_{i}$, then the verification is successful and the post-measurement state $U\ket{q'_{i}}$ is sent back to the authentic party for a future verification. Otherwise, verification is failed/rejected.

\begin{figure}
\begin{center}
\begin{tikzpicture}
\filldraw[fill=red!15] (-1.2,2.5) rectangle (3.2,-0.5);
\filldraw[fill=blue!15] (3.2,2.5) rectangle (7,-0.5);
\draw(0,0) rectangle(2,2);
\node[red] at (-0.15,2.7) {$Generation$};
\node[blue] at (4.2,2.7) {$Verification$};
\node at (1,1) {$\Lambda_{U}$};
\draw(0,1) -- (-1,1) ;
\draw(2,1.5) -- (2.5,1.5);
\draw(2,0.5) -- (3.5,0.5);
\node[above] at (-0.5,1) {$\ket{\psi_{in}}$};
\node[above] at (2.3,1.5) {$m_{i}$};
\node[above] at (2.6,0.5) {$U\ket{m_{i}}$};

\draw(3.5,0) rectangle(5.5,2);
\node at (4.5,1) {$\Lambda_{U}$};
\draw(5.5,1.5) -- (6,1.5);
\node[above] at (5.8,1.5) {$m'_{i}$};
\draw(5.5,0.5) -- (6.5,0.5);
\node[above] at (6.1,0.5) {$U \ket{m'_{i}}$};
\end{tikzpicture}    
\end{center}
\caption{Generation (in red) and verification (in blue) methods for CRPs. Verification is successful if $q_{i} = q'_{i}$ and it is otherwise failed. }
\label{Fig:PUF_construct}
\end{figure}


%

\subsection{Unforgeability Notion for an Ideal Quantum PUF}

For any \emph{ideal QPUF}, the concept of unforgeability lies in the difficulty of replicating its set of queries (CRPs) by any adversary having query access to it.

Let $\mathcal{A}$ be an arbitrary adversary having query database $Q_{\mathcal{A}}$. From \cref{eq: Query Structure}, the structure of $Q_{\mathcal{A}}$ will be,
\begin{equation}
\label{eq: Adversary Query Structure}
    Q_{\mathcal{A}} \coloneqq \{\ket{\psi^{\mathcal{A}}_{i}}, q^{\mathcal{A}}_{i}, U \ket{q^{\mathcal{A}}_{i}}\},
\end{equation}
where $\{\ket{\psi^{\mathcal{A}}_{i}}\}$ are input states initialised by the adversary for making a query. Then, for a security parameter $\lambda$, the adversary is termed as $\lq$\emph{Quantum Exponential Time} (QET)' if $|Q_{\mathcal{A}}| = \exp(\lambda)$, and $\lq$\emph{Quantum Polynomial Time} (QPT)' if $|Q_{\mathcal{A}}| = \text{poly}(\lambda)$. Consider an arbitrary measurement outcome $q \notin Q_{\mathcal{A}}$. Based on the information gained from the query database, let the adversary produce an arbitrary pure state $\ket{\psi_{\mathcal{A}}}$ to pass QPUF verification for the outcome $q$ and let the probability of passing, for the adversary, be $P_{q}$.
Following \cref{eq: Token_prob}, we can write it as
\begin{equation}
\label{eq: Random Variable Definition}
    P_{q} = \Tr \left[\Pi_{q}^{U} \ketbra{\psi_{\mathcal{A}}}{\psi_{\mathcal{A}}} \right].
\end{equation}
Since $U$ is a random variable, $P_{q}$ is also a random variable. Considering the  storage and processing capabilities of the adversary, we can now proceed to define two notions of unforgeability.

\begin{definition}[Quantum Exponential Unforgeability]
\label{Def: Exp_Unforgeability}
Let $|Q_{\mathcal{A}}| = \mathrm{exp}(\lambda)$, where $\lambda$ is the security parameter and $Q_{\mathcal{A}}$ is the query database of an arbitrary adversary, denoted as $\mathcal{A}$. If $\forall$ classical measurement outcomes $q \notin Q_{\mathcal{A}}$ we have
\begin{equation}
    E [P_{q}] = \negl(\lambda),
\end{equation}
then the QPUF is deemed exponentially unforgeable.
\end{definition}

\begin{definition}[Quantum Existential Unforgeability]
\label{Def: Exist_Unf}
Let $|Q_{\mathcal{A}}| = \mathrm{poly}(\lambda)$, where $\lambda$ is the security parameter and $Q_{\mathcal{A}}$ is the query database of an arbitrary adversary, denoted as $\mathcal{A}$. 
If $\forall$ classical measurement outcomes $q \notin Q_{\mathcal{A}}$ we have
\begin{equation}
    E [P_{q}] = \negl(\lambda),
\end{equation}
then the QPUF is deemed existentially unforgeable.
\end{definition}

\subsubsection{Comparison with standard definitions of Existential Unforgeability}

Notice that there is a one-to-one correspondance between the post-measurement states and the classical outcomes. In other words, if we input any post-measurement state $U\ket{q_{i}}$ to the QPUF channel $\Lambda_{U}$ we get the outcome $q_{i}$ always \emph{deterministically}. In general, the existential unforgeability condition requires correct prediction of an output given an input that is different from all inputs in the query database of an adversary. In the way we defined existential unforgeability we require the correct prediction of the input given an output. Since there is a one-to-one correspondence between inputs and outputs, (\textit{i.e. if we consider only the set of post-measurement states and the classical outcomes corresponding to them as domain and range, $\Lambda_{U}$ is invertibile for those choices of sets}) these two notions are mathematically equivalent.

\subsection{Security Results for Ideal QPUF}

The following theorems present the security results for the
ideal QPUF.

\begin{theorem}[No Exponential Unforgeability]
\label{Thm: Exponential Unforgeability}
No Haar random basis measurement based ideal QPUF $\Lambda_{U}$ has exponential unforgeability.
\end{theorem}

\begin{proof}
With an exponential number of queries, any adversary can conduct process tomography on an unknown quantum channel, obtaining a complete description of it. For a quantum channel on a $D$-dimensional Hilbert space, $D^{2}$ queries suffice for standard process tomography \cite{NielsenChuang2010}. Since for a $\lambda$-qubit system $D= 2^{\lambda}$, we can then allow the adversary to make $4^{\lambda}$ queries and thus effectively clone the QPUF.

This completes the proof.
\end{proof}

\begin{theorem}[Existential Unforgeability]
\label{Thm: Existential Unforgeability}
For any Quantum Polynomial Time adversary (QPT) with query database $Q_{\mathcal{A}}$, the expected passing probability in an ideal QPUF verification of any classical measurement outcome $q \notin Q_{\mathcal{A}}$ is bounded as follows
\begin{equation}
\label{eq: Result EXist_unf}
    E[P_{q}] \leq \frac{1}{D - |Q_{\mathcal{A}}|},
\end{equation}
where $P_{q}$ denotes the probability of success for passing the QPUF verification for an outcome $q$ and $D$ is the dimension of the unitary in $\Lambda_{U}$.
\end{theorem}
The proof of this theorem can be found in \cref{Proof: Theorem 2}. For any $\lambda$-qubit
system, the value of $D$ would be $2^{\lambda}$ where $\lambda$ is the security parameter. By definition of QPT, $|Q_{\mathcal{A}}| = \mathrm{poly}(\lambda)$. Hence, from \cref{eq: Result EXist_unf} we have, $\forall q$,
\begin{align*}
    E[P_{q}] &\leq \frac{1}{2^{\lambda} - \mathrm{poly}(\lambda)} \\
    &= \negl(\lambda).
\end{align*}
Thus, \cref{Thm: Existential Unforgeability} implies that any Haar random basis measurement QPUF $\Lambda_{U}$ is \emph{existentially unforgeable}. 


\begin{figure*}
\centering
\includegraphics[width=0.9\textwidth,clip]{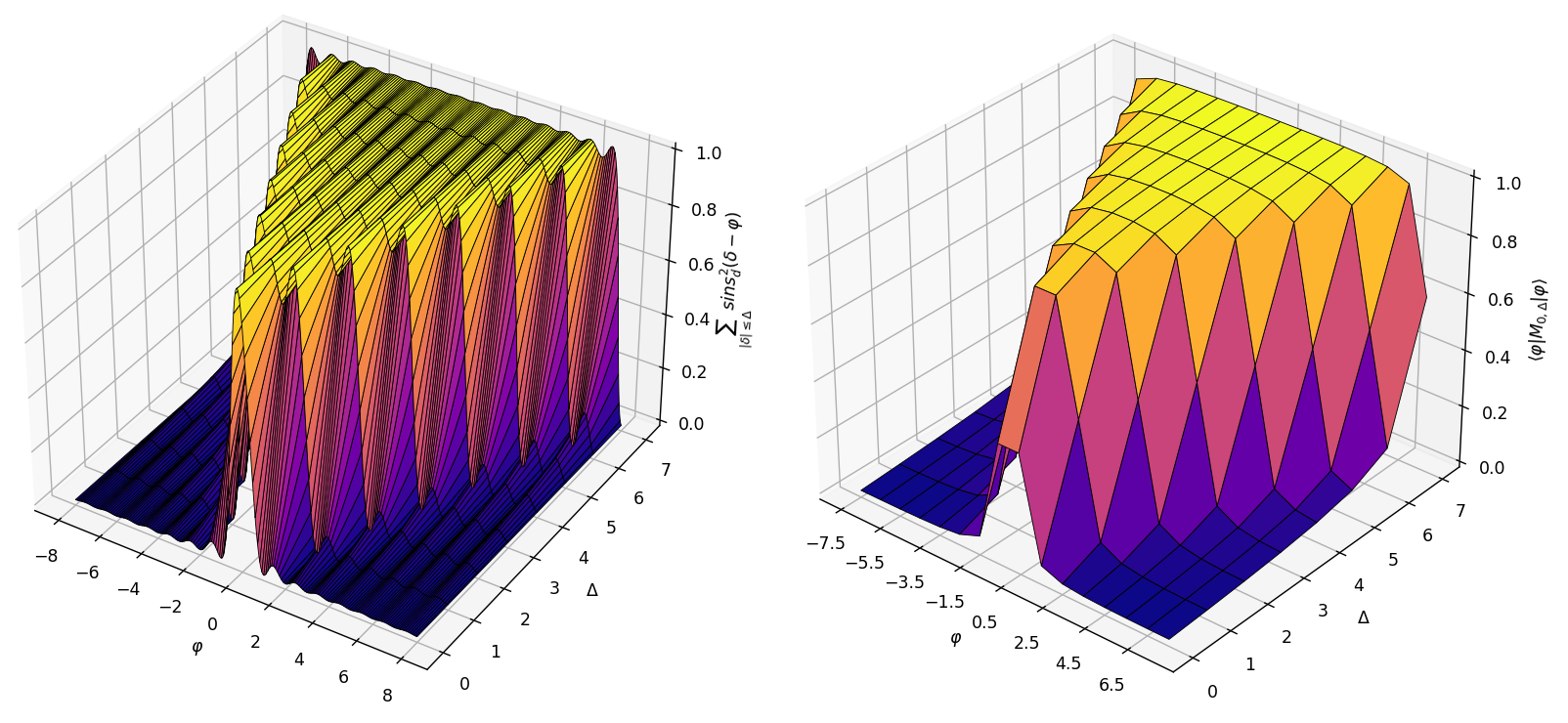}
\caption{Approximation of generated output state to a $sins_{d}^{2}(\delta - \phi)$: \textbf{Left)} Eigenvalue of eigenvector $\ket{\phi}$ for the verification of $M_{0,\Delta}$. \textbf{Right)} Special case where $\phi\in [-1/2, +1/2]$ for $M_{0,\Delta}=\sum_{\abs{\delta}\leq\Delta,j} sins_{d}^{2}\ket{\phi}\bra{\phi}$. The validity of the approximation increases with larger values of $\Delta$.}
\label{ver1}       
\end{figure*} 


\subsection{Implementations}

\subsubsection{Random Control-NOT based QPUF (RCNOT-QPUF)}
\label{RCNOT-QPUF}
Let us assume that, given access to a \emph{Haar Random} unitary $U$, we have oracle access to its inverse, $U^{\dagger}$. We then have an exact implementation of the "\emph{ideal  QPUF}" as the controlled-NOT operation with randomized control followed by a standard basis measurement. The QPUF circuit with outcomes in $\mathbb{Z}_{D}$ can be described as:
    \begin{equation*}
    \begin{quantikz}
   && 
   &\gate{U}&\ctrl{2}  &  \gate{U^{\dagger}} & \qw 
    \\
   \Pi^{U}_{i}= &&          &             &           &             &    &   &   &
    \\
    &&
    \lstick{$\ket{0}$}& \qw & \targ{} & \qw & \meterD{\ket{i}} & \cw & i \; \; ,
    \end{quantikz}
    \end{equation*}
where the C-NOT is defined as
\begin{equation}
    CX = \sum_{i \in \mathbb{Z}_{D}} \ketbra{i}{i} \otimes X^{i}, 
\end{equation}
and where $X$ is the generalised Pauli X gate defined as
\begin{equation}
    X \equiv \sum_{i \in \mathbb{Z}_{D}} \ketbra{i \oplus 1}{i}.
\end{equation}
Here we also make an abuse of notation, as we did in \cref{eq: CU}. This circuit also corresponds to the $D$ Kraus operators of the quantum instrument
\begin{equation}
    \Lambda_{U} (\rho) = \sum_{i \in \mathbb{Z}_{D}} \ketbra{i}{i} \otimes \Pi^{U}_{i} \rho \Pi^{U \dagger}_{i},
\end{equation}
where
\begin{equation}
    \Pi^{U}_{i} = U \ketbra{i}{i} U^{\dagger}.
\end{equation}
This is exactly the Haar random von Neumann measurement channel from \cref{eq: Channel_Def}. In this implementation, a crucial assumption is that the complete quantum instrument $\Lambda_{U}$ constitutes the QPUF, not solely the random unitary. This distinction is critical as it limits adversaries to querying $\Lambda_{U}$ rather than the unitary $U$ itself.

\subsubsection{Phase-Estimation based QPUF (PE-QPUF)}
\label{PE-QPUF}

\textbf{\emph{Hardware Requirements:}}
Since the QPE protocol approximates a von Neumann measurement, it would be an approximate implementation of the \emph{ideal QPUF} in \cref{Hardware_Req}. Again it is important to note the assumption that the entire quantum instrument $\Lambda_{U}^{QPE}$ is the QPUF. We then redefine the hardware requirements as:
\begin{definition}[Approximate QPUF]
Any QPE channel $\Lambda_{U}^{QPE}$ in \cref{eq: QPE Channel} constructed with Haar random unitary $U$ is considered an $\lq$approximate QPUF' if it satisfies the following hardware requirements for some security parameter $\lambda$:
\begin{itemize}
    \label{Hardware_Req QPE}
    \item
    \textbf{Unknownness:} The Kraus operators (\cref{eq:kraus-circuit}) $\{U_{k}\}_{k \in \mathbb{Z}_{d}}$  for the channel $\Lambda_{U}^{QPE}$ remain unknown.
    \item $\delta_{c}$-\textbf{Collision-Resistance:}
$\forall \Delta; 0 \leq \Delta \leq d-1$ , $\exists \delta_{c} ; 1 \geq \delta_{c} >0$ such that for any two different measurement outcomes $i , j$ satisfying $|i - j|\geq \Delta$ , the probability of the post measurement states $\ket{\psi_{i}},\ket{\psi_{j}}$ being different by at least $\delta_{c}$ in trace distance satisfies
\begin{equation}
\Pr \big[D(\ketbra{\psi_{i}}{\psi_{i}},\ketbra{\psi_{j}}{\psi_{j}}) \geq \delta_{c}\big] \geq 1 - \negl(\lambda).
\end{equation}
\item $\delta_{r}$-\textbf{Robustness:} 
$\forall \Delta; 0 \geq \Delta \geq d-1$ , $\exists \delta_{r} ; 1 \geq \delta_{r} >0$ such that for any two measurement outcomes $i , j$ satisfying $|i - j|\leq \Delta$, the probability of the post measurement states $\ket{\psi_{i}},\ket{\psi_{j}}$ being close by at least $\delta_{r}$ in trace distance satisfies
\begin{equation}
\Pr \big[D(\ketbra{\psi_{i}}{\psi_{i}},\ketbra{\psi_{j}}{\psi_{j}}) \leq \delta_{r}\big] \geq 1 - \negl(\lambda).
\end{equation}
\item $\delta_{\mu}$-\textbf{Uniqueness:} 
The probability of any two distinct QPUFs, $\Lambda_{U_{i}}^{QPE}$ and $\Lambda_{U_{j}}^{QPE}$, being distinguishable by at least an amount of $0 \leq \delta_{\mu} \leq 2$ in diamond norm satisfies
\begin{equation}
\label{eq: Uniqueness}
\Pr \left[||\Lambda_{U_{i}}^{QPE} - \Lambda_{U_{j}}^{QPE}||_{\diamond} \geq \delta_{\mu} \right] \geq 1 - \negl (\lambda).
\end{equation}
\end{itemize}
\end{definition}
Notice, that for any \emph{approximate QPUF}, if $\delta_{\mu} = 2$ and if $\forall \Delta > 0$; $\delta_{c} =1$ and $\delta_{r} = 0$, we get an \emph{ideal QPUF}.

\subsubsection{Approximate Authentication Protocol:}

In this section, we describe the authentication protocol with respect to the approximate implementation.

\paragraph{\textbf{Approximate Token Generation Protocol:}}
\label{section: PE-QPUF Helper Data Construction}
On a general input state $\ket{\psi_{in}}$ we apply the Quantum Phase Estimation circuit \cite{Kitaev96}, mentioned earlier in \cref{QPE}, with a Haar random unitary $U$,
    \begin{equation*}
    \begin{quantikz}
    \lstick{$\ket{0}$} &\gate{F^\dagger}&\ctrl{2}  &  \gate{F} & \meterD{\ket{m_{0}}} & \cw & m_{0}  
    \\
    &           &             &           &             &    &   &   &
    \\
    \lstick{$\ket{\psi_{in}}$}& \qw &\gate{U}& \qw & \qw \rstick{$\ket{\psi_{m_{0},\Delta}}$ ,}.
    \end{quantikz}
    \end{equation*}
to obtain a quantum token as the post-measurement state, $\ket{\psi_{m_{0},\Delta}}$, and its classical verifier value as the measurement outcome $m_{0}$.  
  

\begin{figure}
\centering
\includegraphics[width=0.45\textwidth,clip]{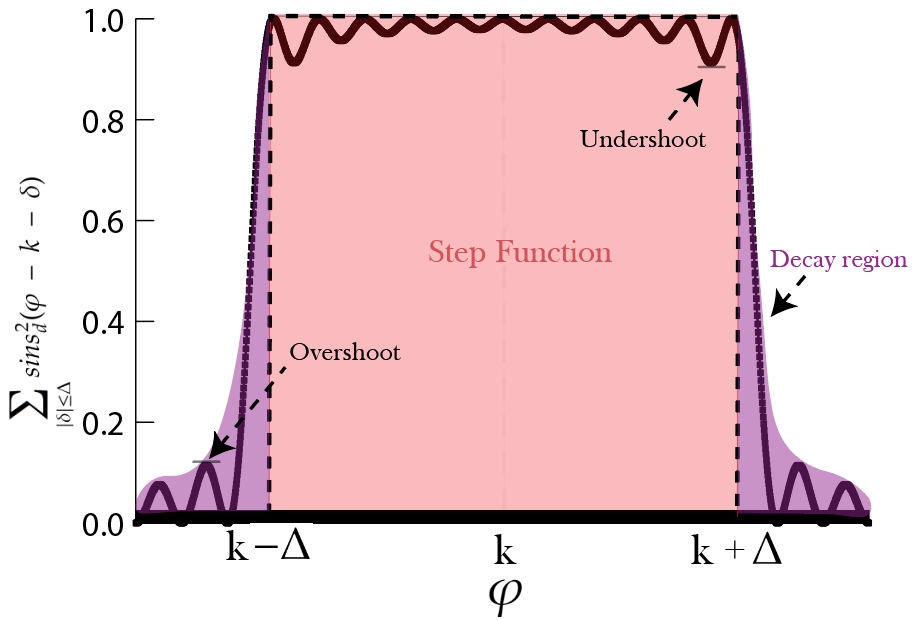}
\caption{Gibbs Phenomenon. Cross-section of \cref{ver1} corresponding to $\Delta=5$, showing step function approximation for the operator $M_{k,\Delta}$ demonstrating its approximate behaviour to the projector $\Pi_{k,\Delta}$ in \cref{eq: Ver_Projectors}. }
\label{Undershoot}  
\end{figure}  
\paragraph{\textbf{Approximate Token Verification Protocol:}}
\label{Subsubsec: QPE_Ver}


Upon verification, the token state $\ket{\psi_{m_{0},\Delta}}$ is requested back by the verifier. The state is then used as input for the same circuit mentioned above, to obtain a pair of measurement outcome and post-measurement state. We consider the $n$-th verification to be successful when the difference between outcomes $m_n$ and $m_{n-1}$ (either from generation, $n-1=0$, or previous verification, $n-1\geq 1$) is small.
    Since the algorithm is not exact, but provides only an estimation, we parameterize the decision boundary by an integer $0 \leq \Delta \leq  \frac{d}{2}$ and say that the verification is successful when the new measurement outcome satisfies
    \[ |\Delta_{n}| \coloneqq |m_n - m_{n-1} | \leq \Delta \mod d.\]
If the $n^{th}$ verification is successful, the classical verifier value is updated from $m_{n-1}$ to $m_{n}$ in the verifier's classical memory and the post-measurement state $\ket{\psi_{m_{n},\Delta}}$ is sent back to the authentic party for a further verification. Due to the commutativity of the Kraus operators in \cref{eq: QPE Kraus Operators}, the conditional probability of having an $n^{th}$ successful verification $P_{m_{n-1},\Delta}$ can be obtained as the expectation value of the operator $M_{m_{n-1},\Delta}$ with respect to the $(n-1)^{th}$ post-measurement state $\ket{\psi_{m_{n-1},\Delta}}$, 
\begin{align}
\label{eq: ISIT n-cond prob. def}
    P_{m_{n-1},\Delta} \equiv \bra{\psi_{m_{n-1},\Delta}}M_{m_{n-1},\Delta}\ket{\psi_{m_{n-1},\Delta}},
\end{align}
where for any $k \in \mathbb{Z}_{d}$ we define,
    \begin{align}
        M_{k,\Delta} 
        &\equiv \sum_{|\delta|\leq \Delta} M_{k + \delta (\mathrm{mod} \hspace{0.1cm} d)} \notag \\
        &
        = \sum_{j\in \mathbb{Z}_{D}} \sum_{|\delta|\leq \Delta} 
        \sins_{d}^{2}(k + \delta - \phi_{j}) \projector{\phi_{j}}
        \\&
        = \sum_{j\in\mathbb{Z}_{D}} \tilde{1}_\Delta(k-\phi_{j}) \projector{\phi_{j}},
        \label{eq:M_k_eigenvalues}
    \end{align}
and where we define
    \begin{equation}
        \tilde{1}_{\Delta}(x)
        =     \sum_{|\delta|\leq \Delta} \sins_{d}^{2}(\delta + x)
        = 1 - \sum_{|\delta| >   \Delta} \sins_{d}^{2}(\delta + x),
    \end{equation}
    which approximates the step function,
    \begin{equation}
        1_{\Delta}(x) = 1_{\Delta}(-x) =
        \begin{cases}
            1 & x \in     [-\Delta,\Delta]
            \\
            0 & x\notin [-\Delta,\Delta]
        \end{cases},
    \end{equation}
    as displayed in \cref{ver1} for $k=0$ ($x=\phi$).
    Recall that the Fourier transform of such a step function is the squared $\sinc$ function. Hence, it should not be a surprise that the $\sins$ series approximates $1_{\Delta}(k-\phi)$, as it approximates a $\sinc$ series for large ancilla dimension $d$. This is known in literature as $\lq$\emph{Gibbs Phenomenon}' \cite{Gibbs} which is dispayed in \cref{Undershoot}.
    
    Let
    \begin{equation}
       \label{eq: Ver_Projectors} \Pi_{k,\Delta} \equiv \sum_{\substack{j \in \mathbb{Z}_{D} \\
       |\phi_{j} - k| \leq \Delta}} \projector{\phi_{j}}.
    \end{equation}
Thus, intuitively, $M_{k,\Delta}$ approximates the projector $\Pi_{k,\Delta}$, which projects the post-measurement state in a $\Delta$ neighbourhood of $k$ as seen in \cref{Undershoot}.
\begin{definition}[$(d,D,\Delta)$-PE-QPUF]
Any QPE channel $\Lambda_{U}^{QPE}$ from \cref{eq: QPE Channel} with aniclla dimension $d$, Haar random Unitary dimension $D$ and decision boundary $\Delta$ satisfying the hardware requirements of \emph{approximate QPUF}
is defined as $(d,D,\Delta)$-PE-QPUF.
\end{definition}

\begin{theorem}[Authentic Verification Probability]
\label{Authentic Verifiation Probability}
For any $(d,D,\Delta)$-PE-QPUF with $\Delta \geq 2$, the conditional probability for the $n^{th}$ successful verification $P_{m_{n-1},\Delta}$, with $m_{n-1}$ as the measurement outcome at the $(n-1)^{th}$ round, can be lower bounded as
\begin{align}
\label{eq:Ver_lowerbound}
    P_{m_{n-1},\Delta} &> \left(1-  \sqrt{1 - f(\Delta)}\right)^{2} \; \; ,
\end{align}
where
\begin{align}
   f(\Delta) \equiv \left(1 - \frac{2}{\pi^{2}\left(\sqrt{\Delta} + \frac{1}{2}\right)}\right) \cdot \left(1 - \frac{2}{\pi^{2}(\Delta - \frac{1}{2})}\right). \notag 
\end{align}
\label{eq:lowerbound}
\end{theorem}
The proof of this theorem can be found in \cref{Proof: Theorem 3}. From \cref{eq:Ver_lowerbound}, for any $(2^{\alpha + \lambda},2^{\alpha + \lambda},2^{\lambda})$-PE-QPUF where $\alpha,\lambda >1$, we have $P_{m_{n-1,\Delta}} > 1 - \negl(\lambda)$. In general we observe that, with higher values of $\Delta$, the verification probability approaches $1$, whereby giving us a closer approximation to a von Neumann measurement. This is not a surprise, as increasing values of $\Delta$ reduces the precision of the phase measurement in the QPE protocol. 
\begin{definition}(Verification rate)
Considering $N_{\mathrm{total}}$ tokens sent to an authentic party, if $N_{\mathrm{pass}}$ pass the verification, we define the verification rate as $0 \leq v_{R} \leq 1$, where,
\begin{equation}
    v_{\mathrm{R}} \equiv \frac{N_{\mathrm{pass}}}{N_{\mathrm{total}}}.  \label{eq:verification rate}
\end{equation}
\end{definition}
For an overall successful authentication it is then required that,
\begin{equation}
    v_{\mathrm{R}} \geq \left(1-  \sqrt{1 - f(\Delta)}\right)^{2}.
\end{equation}
\begin{figure*}
\centering
\includegraphics[width=0.9\textwidth,clip]{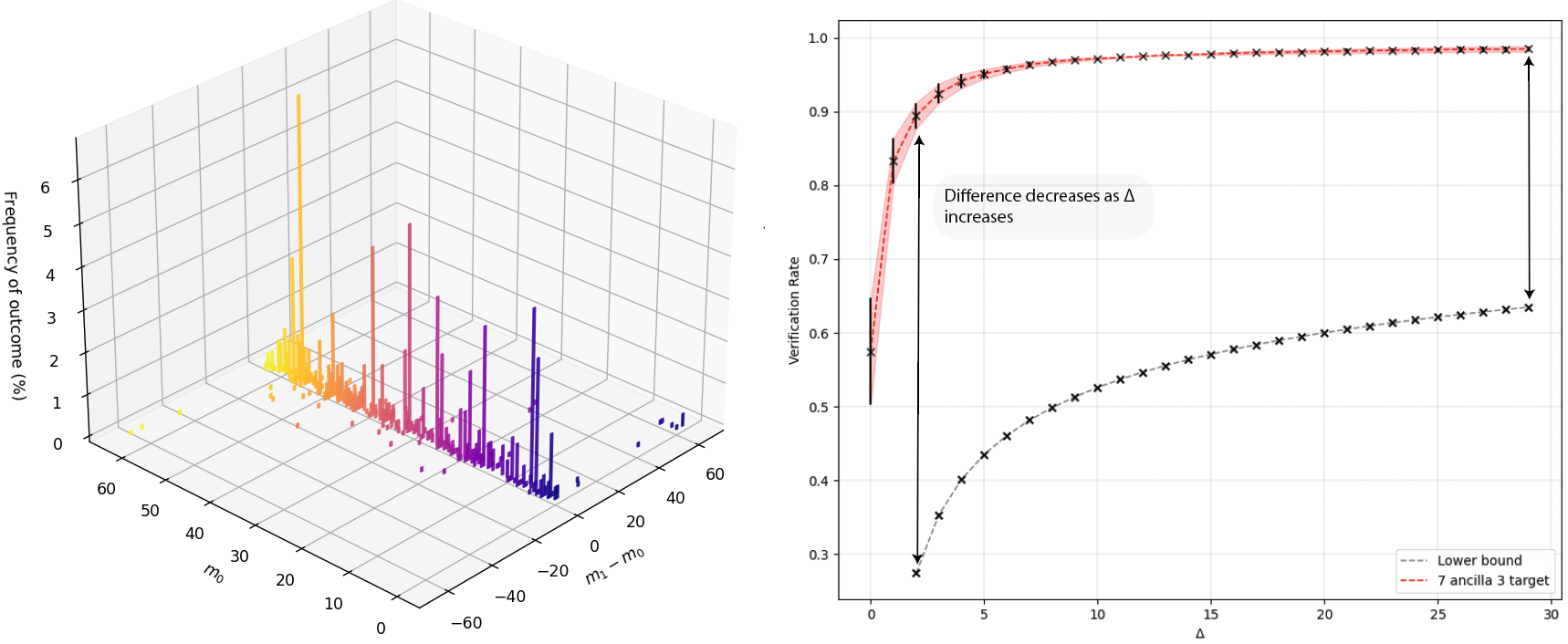
}
\caption{On the left, a $(2^{6},2^{6},\Delta)$-PE-QPUF is simulated with $10^{3}$ quantum states. $m_{0}$ and $m_{1}$ represent measurement outcomes at generation and $1^{st}$ verification, respectively. As expected, high peaks of outcome frequency are observed around low values of $m_{1} - m_{0}$ ($\Delta$). On the right, the simulation for a $(2^{7},2^{3},\Delta)$-PE-QPUF with $10^{4}$ states compares the analytical lower bound with the verification rate from \cref{eq:verification rate}.The red-shaded region represents the standard deviation over 5 iterations with $10^{4}$ states each. The difference between simulated results and the lower bound decreases with higher $\Delta$ values.}
\label{Fig:Multiple_ver}    
\end{figure*}
To validate our analytical predictions on $\Delta$ and $v_{R}$, we simulated the verification protocol for $(2^{6},2^{6},\Delta)$-PE-QPUF and $(2^{7},2^{3},\Delta)$-PE-QPUF on the IBM aer-simulator backend, with $10^{3}$ and $10^{4}$ states respectively, as shown in \cref{Fig:Multiple_ver} (left) and \cref{Fig:Multiple_ver} (right) respectively. The Haar random unitary was simulated pseudo-randomly via QR decomposition \cite{qr-decomp}.

\subsubsection{Security for PE-QPUF:}
\label{Section: QPE Security}

Let the adversary produce an arbitrary pure state $\ket{\psi_{\mathcal{A}}}$ to forge the QPUF verification. From \cref{eq: ISIT n-cond prob. def}, for any measurement outcome $m$, the probability of success for the adversary will be
\begin{equation}
    P_{m,\Delta} \equiv \bra{\psi_{\mathcal{A}}}M_{m,\Delta} \ket{\psi_{\mathcal{A}}}.
\end{equation}
Since $M_{m,\Delta}$ is random, $P_{m,\Delta}$ is also a random variable.
\begin{theorem}[$2\Delta$-Existential Unforgeability]
\label{Theorem: PE-QPUF Existential Unforgeability}
Consider a QPT adversary with query database $Q_{\mathcal{A}}$ and a $(2^{\alpha + \lambda},2^{\alpha + \lambda},2^{\lambda})$-PE-QPUF with $\alpha,\lambda > 1$. Let $m \notin Q_{\mathcal{A}}$ be any measurement outcome such that $\forall k \in Q_{\mathcal{A}}$; $|m-k| \geq 2 \Delta$. The  expected probability of the adversary for a successful verification of the outcome $m$ would be bounded as follows:
\begin{align}
    & \E[P_{m,\Delta}] \notag \\
    &\leq \left(1 - \frac{2}{\pi^{2}(\frac{\Delta}{2} - 1)}\right)\cdot \left(\frac{1}{\frac{d}{2(\Delta+\frac{\Delta}{2})} - |Q_{\mathcal{A}}|}\right) + \frac{2}{\pi^{2}(\frac{\Delta}{2} - 1)} \notag \\
   & = \left(1 - \frac{2}{\pi^{2}(2^{\lambda - 1} - 1)}\right)\cdot \left(\frac{1}{\frac{2^{\alpha}}{2(1 + \frac{1}{2}) }- |Q_{A}|}\right) + \frac{2}{\pi^{2}(2^{\lambda - 1} - 1)} \notag \\
    &= \negl(\alpha, \lambda).
\end{align}
\end{theorem}
The proof of this theorem can be found in \cref{Proof: Theorem 4}.

\section{Discussion and Future Work}
\label{Discussion}

We begin the discussion, by stating that in our QPUF models we do not address \lq man-in-the-middle' attacks, and we explain our rationale for this. We emphasize that such attacks are a distinct issue, independent of the QPUF framework, and should be treated separately, which is why they fall outside the scope of this paper. Specifically, man-in-the-middle scenarios can be viewed as arbitrarily varying channels (AVCs) established between communicating parties. As in other QPUF models \cite{Skori2010, Arapinis2021quantumphysical,SWAP_Mina},communication implicitly assumes that these AVCs have positive capacity. For instance, standard QPUF protocols are also vulnerable to "jamming" attacks, where an adversary might intercept and alter the challenge state or steal the legitimate responses, preventing successful verification or enabling false verification respectively. These issues can be addressed by studying AVCs with positive capacities and integrating them into our QPUF scheme.

Thus, omitting man-in-the-middle attacks does not detract from the validity of our proposed models. Nevertheless, incorporating such attacks would enhance the practicality of our approach, and we plan to explore this in future work.

\subsection{Comaprison between QR-PUF and our proposed models}
The (QR)-PUF model \cite{Skori2010} faces several challenges that are absent in the model presented in \cite{Arapinis2021quantumphysical} and also our proposed models:
\begin{itemize}
    \item The Challenge-Response Table (CRT) consists primarily of classical strings, meaning the potential of quantum states isn't fully utilized.
    \item The PUF's unitary is fully known, allowing an adversary with a quantum computer to simulate it and predict the response states.
    \item There is a trade-off between noise resistance and unclonability in the (QR)-PUF model.
    \item The certifier has access to the CRT, making them a trusted party.
\end{itemize}

On the other hand, a key requirement in our models, as well as in the model from \cite{Arapinis2021quantumphysical}, is the assumption of Haar random unitaries. Since perfect Haar randomness is unattainable, this remains a potential limitation until the security of our schemes is evaluated under weaker assumptions. However, progress has been made with unitary unitary designs \cite{Haferkamp2022randomquantum, boche2019simultaneous, boche2017randomness}, and analyzing the security of our schemes using unitary-designs is an interesting direction for future work. If the adversary is restricted to only query access of the QPUF device, it was shown in \cite{doosti2022connection}, that PRUs can be used as QPUFs.

Another promising avenue is exploring many-body quantum systems that can generate states with behavior close to Haar randomness. Investigating how such systems can be leveraged for our proposed models' security analysis will also be an important part of future research. These areas are left for future work and are beyond the scope of this paper.

\subsection{Comparison between Hybrid PUF \cite{CDMWAK23} and our proposed models}

While developing our model, we noticed that the concept of combining classical elements with quantum constructions was introduced earlier through the Hybrid PUF \cite{CDMWAK23}. A Hybrid PUF consists of two main components: a Classical PUF (CPUF) and a public encoder that maps the CPUF's outputs to quantum states. The key assumption in this model is that the CPUF outputs are hidden from the adversary. Security is based on the difficulty for an adversary to infer the CPUF's outputs from the encoded quantum states. However, this reliance on the secrecy of the CPUF's outputs makes the Hybrid PUF model weaker than ours, which requires no such assumption. 

In the following subsections, we mention some limitations of the implementations of the $\lq$\textit{Ideal QPUF}' and propose possible interesting future research directions.

\subsection{Limitations for RCNOT-QPUF}

This model assumes that, when provided with a Haar random unitary $U$, the QPUF device can carry out its inverse operation $U^{\dagger}$. However, the efficiency of performing such an operation remains uncertain. 

\subsubsection*{Open research direction:}

Attempts have been made to develop algorithms for computing the inverse of an unknown $D$-dimensional unitary \cite{sardharwalla2016universal,Sedl_k_2019,Quintino_2019_Prob,Quintino_2019,Dong_2021,Quintino2022deterministic}. None of the cited algorithms are deterministic and exact at the same time. It is also shown, in the above-mentioned references, that with the increase in precision of the implementation (\textit{deterministic approximate case}) or increase in the probability of correct implementation (\textit{probabilistic exact case, \cite{Sedl_k_2019,Quintino_2019_Prob}}) the gate cost complexity of the algorithms increases exponentially. Thus, although the implementations work, they are inefficient. Recently, an attempt was made for finding $\lq$\emph{deterministic}' and $\lq$\emph{exact}' implementation \cite{Yoshida_2023}. However, the introduced algorithm works only for $SU(2)$ unitaries and it remains an open problem to extend the algorithm to general dimensions and finding a closed form expression for its complexity as a function of dimension. To summarize, it is an open problem to find a computationally efficient method for inverting any unknown blackbox unitary.

\subsection{Limitations for PE-QPUF}

As the model is grounded in Quantum Phase Estimation, it inherits the \emph{exponential overhead} issue associated with it. For any general Unitary matrix $U$, Quantum Phase Estimation necessitates the implementation of the gate $CU$ $d-1$ times (for qubits, and in general scales as $d$), where $d$ is the dimension of the ancilla system. For a $\lambda$ qubit ancilla system, the gate $CU$ needs to be applied $2^{\lambda} - 1$ times. The gate cost complexity for such implementation can be computed \cite{NielsenChuang2010}, by summing a geometric series, to be $2^{\lambda}$. Other schemes, like \emph{fast phase estimation} \cite{svore2013faster}, performs phase estimation with logarithmic circuit depth, but the classical post processing remains exponentially costly.  
Consequently, a \emph{trade-off} between $\lq$\emph{implementability}' and $\lq$\emph{security}' emerges, as efficient implementation requires smaller dimensions, while enhanced security needs larger dimensions.

\subsubsection*{Open research direction:}

Since implementing $CU^{2^{\lambda} - 1}$ is exponentially costly (gate cost complexity) in general, one can look for classes of random unitary matrices (other probability measures than the Haar measure) for which the implementation is secure and polynomially costly. 



Another proposal is to convert the multiple applications of the gate $U$ into a single application. Recall from Schrodinger's equation that any unitary $U$ is nothing but the time evolution operator obtained by exponentiating a physical Hamiltonian, i.e., for any Hamiltonian $H$ we can obtain an unitary $U$ as
\begin{equation}
    U(t) = e^{\frac{i}{\hbar}H t},
\end{equation}
where the operator $U(t)$ evolves a state for a time $t$. $\beta$ applications of $U$ would then result into
\begin{equation*}
    U^{\beta}(t) = e^{\frac{i}{\hbar}\beta H t}.
\end{equation*}
Consider a new Hamiltonian $H' = \beta H$. We can then obtain the operator $U'$ by performing
\begin{equation}
    U'(t) = e^{\frac{i}{\hbar}H' t} = U^{\beta}(t).
\end{equation}
If we could create all the Hamiltonians $H'$, for integer powers $\beta$ of $U$, we can then produce a collection of $\lambda$ many gates
\begin{equation}
    \{U^{i
    }\}_{i = 0}^{\lambda-1}.
\end{equation}

Using these gates we can implement $CU^{2^{\lambda - 1}}$ with gate cost complexity $\lambda$. A possibility of creating $H' = \beta H$ is to absorb $\beta$ into the coupling parameters of $H$ by tuning them with external interactions (thermal, magnetic..).
Thus, there is interesting future work in finding many body systems where such tunability is possible. For example, consider the Heisenberg-XYZ Hamiltonian 
\begin{equation}
    \hat{H}=-\frac{1}{2}\sum_{j=1}^{N}(J_{x_{j}}\sigma_j^x\sigma_{j+1}^x+J_{y_{j}}\sigma_j^y \sigma_{j+1}^y+J_{z_{j}}\sigma_j^z \sigma_{j+1}),
  \end{equation}
where $J_{x_j},J_{y_j},J_{z_j}$ are the couplings. If the couplings were random, we would have a random unitary constructed out of such a Hamiltonian. Now, if tuning the couplings with external interactions was possible, this could be a suitable candidate. It is also important to note that such an implementation is not solving the problem but rather translating $\lq$\emph{time-complexity}' into $\lq$\emph{energy-complexity}' as for tuning the couplings one has to access higher energy states. However, this translation of the problem from \emph{time} to \emph{energy} may be sometimes useful for practical purposes as relatively higher energies might be more accessible than exponentially long waiting times for token generation. 

In the following section we present a comparison between the two models proposed in this paper.

\subsection{Comparision between MB-QPUF and Ideal QPUF}

In this section, we compare the two models presented in this paper. For near-term practical implementations, we argue that the MB-QPUF is more suitable than the Ideal QPUF for the following reasons:
\begin{itemize}
    \item The MB-QPUF can be implemented efficiently with oracle access to a Haar random unitary, whereas in the proposed implementations in this paper, the Ideal QPUF requires an exponential gate cost, making it less feasible.
    \item Storing entangled states is not a concern for the \textit{Ideal QPUF}, as most constructions in the literature \cite{Arapinis2021quantumphysical,SWAP_Mina,Skori2010} already involve storing highly entangled states.
    \item In practical scenarios, the prover typically has limited resources. The MB-QPUF does not require quantum memory on the prover's side, whereas the Ideal QPUF does. Since quantum memory is currently expensive, this is an advantage. However, this benefit is somewhat reduced by the fact that possessing a QPUF device itself can be costly in the near term.
\end{itemize}
That said, if the challenges around quantum memory are resolved and we develop an efficient implementation of the Ideal QPUF, it would offer more advantages than the MB-QPUF. This is because the Ideal QPUF is a complete device that even the verifier cannot tamper with, while in the MB-QPUF, the query process can potentially be compromised.

An efficient implementation of the \textit{ideal QPUF} also addresses the long-standing open problem of public-key quantum money schemes \cite{aaronson2012quantummoneyhiddensubspaces}. Here's a brief overview of how it works:
\begin{itemize}
    \item A public server (or mint) possesses the ideal QPUF device $\Lambda_{U}$.
    \item The mint initialises arbitrary input states and measures them using $\Lambda_{U}$ to generate money bills $U\ket{i}$ and corresponding classical outcomes $i$. These classical outcomes will serve as serial numbers which are printed on the money bill. Then these money bills are distributed.
    In the money generation process there will be other details like what happens when two money bills have the same serial number during production? Two money bills with the same measurement outcome are discarded. Notice, that the mint wants to print polynomially many money bills because it makes no sense if the money generation scheme has exponential complexity. However, there are exponentially many measurement outcomes possible with uniform probability on average. It is therefore trivial to show that the collision probability is negligible.
    \item When Alice transfers a bill to Bob, Bob sends it to the server for verification. If successful, the server returns the post-measurement state to Bob. Notice, that if the verification is passed successfully the measurement does not change the quantum state of the money bill.
\end{itemize}

\section{Methods}

\subsection{Proof of \cref{Thm: Existential Unforgeability}}
\label{Proof: Theorem 2}

\begin{proof}(\cref{Thm: Existential Unforgeability})
We begin the proof by considering an arbitrary measurement outcome $q \notin Q_{\mathcal{A}}$, where $Q_{\mathcal{A}}$ is the query database of the adversary.

Following \cref{eq: Random Variable Definition}, we write the passing probability for the adversary as
\begin{equation}
    P_{q} = \Tr \left[\Pi_{q}^{U} \ketbra{\psi_{\mathcal{A}}}{\psi_{\mathcal{A}}} \right],
\end{equation}
where $\ket{\psi_{\mathcal{A}}}$ is the guess state of the adversary. 

Any unitary $U$ can be written in standard basis as
\begin{equation}
     U = (U \ket{0}, U\ket{1}, \cdots, U \ket{D-1}).
\end{equation}
Let us now consider the following construction for a Haar random unitary \cite{petz2003_HaarConstruction}, which is inspired by the Gram-Schmidt orthogonalisation process \cite{Axler_LinAlg}:
\begin{itemize}
    \item Choose $i^{th}$ column $U\ket{i}$ uniformly at random from $\mathbb{C}^{D}$.
    \item Choose $j^{th}$ column $U\ket{j}$ uniformly at random from the $D-1$ dimensional orthogonal subspace of $U\ket{i}$.
    \item Choose $k^{th}$ column $U\ket{k}$ uniformly at random from the $D-2$ dimensional orthogonal subspace of the subspace spanned by $U\ket{i}$ and $U\ket{j}$.
    \item Keep iterating until all the columns are obtained.
\end{itemize}
By exploiting the left invariance of the Haar measure, it is shown in \cite{petz2003_HaarConstruction} that such a construction gives a Haar random unitary. Recall the structure of the query database. $Q_{\mathcal{A}}$: 
\begin{equation*}
        Q_{\mathcal{A}} \coloneqq \{\ket{\psi_{\mathcal{A}}^{i}} , q_{i} , U\ket{q_{i}}\}_{i \in \mathbb{Z}_{|Q_{\mathcal{A}}|}},
    \end{equation*}
where we let $\{\ket{\psi_{\mathcal{A}}^{i}}\}$ be the set of input states initialised by the adversary to query the QPUF. Since a quantum measurement can collapse two different states onto a same state, there can be degeneracies in $Q_{\mathcal{A}}$, i.e., for two states $\ket{\psi_{\mathcal{A}}^{i}} \neq \ket{\psi_{\mathcal{A}}^{j}}$ we may have,
\begin{equation*}
    (q_{i},U\ket{q_{i}}) = (q_{j},U\ket{q_{j}}).
\end{equation*}
Consider $Q'_{\mathcal{A}}$, with $|Q'_{\mathcal{A}}| = |Q_{\mathcal{A}}|$, having no degeneracies. It is then evident that such a case is the best case scenario for the adversary, as in such case the maximum number of different token states can be owned. Hence, the success probability, with query database $Q'_{\mathcal{A}}$, will always be greater than or equal to the one with $Q_{\mathcal{A}}$. We would now perform a calculation with $Q'_{\mathcal{A}}$ to find an upperbound on the expected success probability. Subsequently, consider $q \notin Q'_{\mathcal{A}}$. Passing the verification for such $q$ is equivalent to guessing the state vector $U\ket{q}$. It follows from the construction of Haar random unitary above \cite{petz2003_HaarConstruction} that the task of guessing $U\ket{q}$ is equivalent to guessing a vector sampled uniformly at random from a $\mathbb{C}^{D - |Q'_{\mathcal{A}}|}$ dimensional subspace. Now, consider a Haar random unitary $W \in U(D - |Q'_{\mathcal{A}}|)$. By the arguments above, we have
\begin{equation}
    U \projector{m} U^{\dagger} \sim W \projector{0} W^{\dagger},
\end{equation}
where the symbol $\sim$ denotes \emph{"distributed as"}. 

Let us construct a random variable $X$ as:
\begin{equation}
    X \coloneqq \Tr(W \ketbra{0}{0} W^{\dagger} \rho_{\mathcal{A}}).
\end{equation}
It follows from our arguments that
\begin{align}
    X &\sim P_{q} \notag \\
    \implies E[X] &= E[P_{q}].
\end{align}
Thus, we have
\begin{align}
    E[X] &= \int_{W} \Tr(W \ketbra{0}{0} W^{\dagger} \ketbra{\psi_{\mathcal{A}}}{\psi_{\mathcal{A}}}) dW \notag \\
    &= \Tr \left(  \frac{\mathbb{I}_{D-|Q'_{\mathcal{A}}|}}{D-|Q'_{\mathcal{A}}|}  \cdot \ketbra{\psi_{\mathcal{A}}}{\psi_{\mathcal{A}}}\right) \notag\\
    &= \frac{1}{D - |Q'_{\mathcal{A}}|} \notag \\
        &= \frac{1}{D - |Q_{\mathcal{A}}|} \notag \\
        &= E[P_{q}].
\end{align}
Since we have calculated this for the best case query scenario, $Q'_{\mathcal{A}}$,  we in general have
\begin{equation}
    E[P_{q}] \leq \frac{1}{D - |Q_{\mathcal{A}}|}.
\end{equation}
This completes the proof.
\end{proof}
\subsection{Proof of \cref{Authentic Verifiation Probability}}
\label{Proof: Theorem 3}
Before we begin with the proof of \cref{Authentic Verifiation Probability}, we will first introduce two lemmas.
\begin{lemma}[Lower bound for sinc series]
\label{Lemma: Sinc lower bound}
For some decision boundary $\Delta$, an arbitrary measurement outcome $k$ and $\forall x$ such that $|x-k| \leq \Delta - b$, we have the following lower bound on the normalized sinc series for some arbitrary $0 \leq b < \Delta$

\begin{equation}
    \sum_{\substack{|\delta| \leq \Delta \\ 
    |x - k| \leq \Delta - b}} \sinc^{2}(x - k - \delta) \geq 1 - \frac{2}{\pi^{2}(b + \frac{1}{2})}.
\end{equation}
\end{lemma}

\begin{proof}
Without loss of generality, we work with the measurement outcome $k=0$. Also note that, $\forall x$, 
\begin{align}
    \sum_{|\delta| \leq \infty} \sinc^{2}(x - \delta) &= 1 \notag \\
   \implies \sum_{|\delta| \leq \Delta} \sinc^{2}(x - \delta)  + \sum_{|\delta| > \Delta} \sinc^{2}(x - \delta) &= 1.
\end{align}
Thus, we have
\begin{align}
     \sum_{|\delta| \leq \Delta} \sinc^{2}(x - \delta) \geq 1 - \max_{x} \sum_{|\delta| > \Delta} \sinc^{2}(x - \delta).
\end{align}
Now, for $|x| \leq \Delta - b$, we proceed to obtain
\begin{align}
     &\max_{|x| \leq \Delta - b}\sum_{|\delta| >   \Delta}  \textrm{sinc}^{2}(\delta - x) \notag \\ 
     &= \max_{|x| \leq \Delta - b} \left( \frac{2 \sin^{2}(\pi(x))}{\pi^{2}} \lim_{N \to \infty}\sum_{\delta = \Delta +1}^{N} \frac{1}{(\delta - x)^{2}}\right) \notag \\
    &[\because sin^{2}(x)\; \textrm{is periodic in $\pi$} ] \notag \\
    &\leq \max_{|x| \leq \Delta - b} \lim_{N \to \infty} \int_{\Delta + 0.5}^{N +0.5} \frac{2 \sin^{2}(\pi(x))}{\pi^{2}(\delta - x -0.5)^{2}} d\delta  \label{eq: integral bound}\\
    &= \max_{|x| \leq \Delta - b} \lim_{N \to \infty} \left( \frac{2 \sin^{2}(\pi(x))}{\pi^{2}(\Delta + 0.5 - x)} - \frac{2 \sin^{2}(\pi(x))}{\pi^{2}(N +0.5 - x)} \right) \notag\\
    &= \max_{|x| \leq \Delta - b} \left( \frac{2 \sin^{2}(\pi(x))}{\pi^{2}(\Delta + 0.5 - x)} \right) \notag \\
    &\leq   \frac{2 }{\pi^{2}(b + \frac{1}{2})}.
\end{align}
In \cref{eq: integral bound} we have bounded the quadratic sum using an integral, which becomes evident from the illustration \cref{fig: Integral bound}
\begin{figure}
   \begin{center}
    \begin{tikzpicture}[yscale=3]
    \draw[->] (0,0) +(-1,0) -- ++(7,0) node[right] {$x$};
    \draw[->] (0,0) +(0,-.5) -- ++(0,1) node[above] {$y$};
    \foreach \i in {1,2,3,4,5}
    \draw (\i,.05) -- ++(0,-.1) node [below] {$\i$};
    \draw[scale=1, domain=1.5:5, smooth, variable=\x, blue] plot ({\x}, {1/\x/\x});
    \draw[scale=1, domain=1.5:5, smooth, variable=\x, dashed, green]
    plot ({\x}, {1/(\x-0.5)/(\x-0.5)});
    \draw[scale=1, domain=1.3:5, smooth, variable=\x, dashed, red]
    plot ({\x}, {1/(\x+0.5)/(\x+0.5)});
    \draw[green, dashed] (5,1 ) -- ++(1,0) ++(0.2,0) node[right] {$1/(n-0.5)^2$};
    \draw[blue ,       ] (5,.8) -- ++(1,0) ++(0.2,0) node[right] {$1/n^2$};
    \draw[red  , dashed] (5,.6) -- ++(1,0) ++(0.2,0) node[right] {$1/(n+0.5)^2$};
    
    \foreach \x in {2,3,4,5}
    \draw (\x-0.5, 0) -- ++(0,{1/\x/\x}) -- ++(1,0) --  ++(0,{-1/\x/\x});
    \end{tikzpicture}
    \end{center}
    \caption{Integral bound for quadratic sum.}
    \label{fig: Integral bound}
\end{figure}
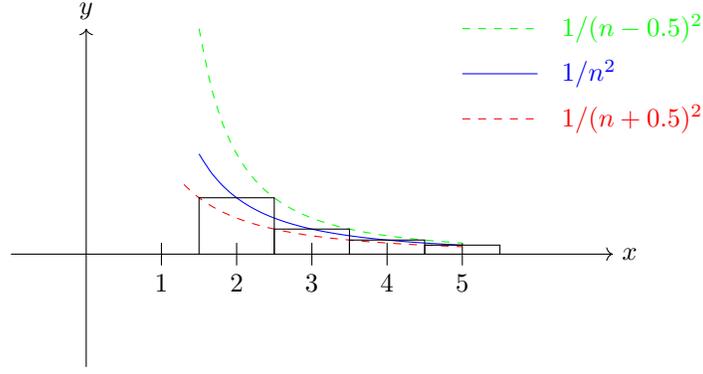
\begin{equation}
    \int_{a}^{b} \frac{dx}{(x + 0.5)^{2}} \leq \sum_{x = a}^{b} \frac{1}{x^{2}} \leq \int_{a}^{b} \frac{dx}{(x - 0.5)^{2}}.
\end{equation}

Therefore, we have
\begin{equation}
    \sum_{\substack{|\delta| \leq \Delta \\ 
    |x-k| \leq \Delta - b}} \sinc^{2}(x - k - \delta) \geq 1 - \frac{2}{\pi^{2}(b + \frac{1}{2})}.
\end{equation}
This completes the proof.
\end{proof}
\begin{lemma}[Upper bound for normalized sinc series]
\label{Lemma: Sinc Upper bound}
For some decision boundary $\Delta$, some arbitrary $c > 1$, an arbitrary measurement outcome $k$ and $\forall x$ such that $|x - k| \geq \Delta+c$, we have the following upper bound on the normalized sinc series

\begin{equation}
    \sum_{\substack{|\delta| \leq \Delta \\ 
    |x - k| \geq \Delta + c}} \sinc^{2}(x - k - \delta) < \frac{2}{\pi^{2}(c - 1)}.
\end{equation}
\end{lemma}

\begin{proof}
The proof of this lemma follows along the similar lines of argument as \cref{Lemma: Sinc lower bound}.
\end{proof}
Now, we proceed to prove \cref{Authentic Verifiation Probability}.
\begin{proof}(\cref{Authentic Verifiation Probability})
For a measurement outcome $m_{n}$ at the $n^{th}$ round, let $\ket{\psi_{m_{n},\Delta}}$ denote the post-measurement state for some decision boundary $\Delta$.

Let $X$ be a random variable, defined as  $X \coloneqq [0,\cdots,d-1]$, denoting the possible measurement outcomes at the $(n-1)^{th}$ round. Let us consider an arbitrary post-measurement state at the $(n-2)^{th}$ round, which can be expressed in the eigenbasis $\{\ket{\phi_{i}}\}_{i = 0}^{D-1}$ of the Haar random unitary $U$ as
\begin{equation}
    \ket{\psi_{m_{n-2},\Delta}} \equiv \ket{\psi} = \sum_{i \in \mathbb{Z}_{D}} c_{i} \ket{\phi_{i}},
\end{equation}
where $c_{i}$'s are complex numbers such that $\sum_{i \in \mathbb{Z}_{D}} |c_{i}|^{2} = 1$. The p.d.f. for $X$ is then obtained as,
\begin{equation}
    \Pr(X = m_{n-1}) = \bra{\psi} M_{m_{n-1}} \ket{\psi}.
\end{equation}
We then define another random variable $P_{X,\Delta}$  that denotes the conditional probability of having an $n^{th}$ successful verification, conditioned on $n-1$ prior successful verifications. The values for the random variable $P_{X,\Delta}$, $\forall m_{n-1} \in \mathbb{Z}_{d}$ is obtained as
\begin{equation}
     P_{m_{n-1},\Delta} = \Tr [M_{m_{n-1},\Delta} \frac{U_{m_{n-1}}\ketbra{\psi}{\psi}U_{m_{n-1}}^{\dagger}}{\bra{\psi} M_{m_{n-1}}\ket{\psi}}].
\end{equation}
The expected conditional probability will be:
\begin{align}
    &\E [P_{X,\Delta}]  
    = \sum_{k \in \mathbb{Z}_{d}} P_{k,\Delta} \cdot \Pr(X = k) \notag \\
    &= \sum_{k\in \mathbb{Z}_{d}} \frac{\Tr[M_{k,\Delta} U_{k} \ketbra{\psi}{\psi}U_{k}^{\dagger}]}{\bra{\psi}M_{k}\ket{\psi}} \cdot \bra{\psi}M_{k}\ket{\psi} \notag \\
    &= \sum_{k \in \mathbb{Z}_{d}} \bra{\psi}M_{k,\Delta}M_{k} \ket{\psi} \notag \\
    &= \sum_{i \in \mathbb{Z}_{D}} |c_{i}|^{2}\sum_{k \in \mathbb{Z}_{d}} \sins_{d}^{2}(\phi_{i}-k)\sum_{|\delta|\leq \Delta} \sins_{d}^{2}(\phi_{i} -k - \delta) \notag \\
    &\geq \min_{\phi_{i},\Delta} \left(\sum_{k \in \mathbb{Z}_{d}} \sinc^{2}(\phi_{i}-k)\sum_{|\delta|\leq \Delta} \sinc^{2}(\phi_{i} -k - \delta)\right) \cdot 1. \label{eq: pre-bound consistency}
\end{align}
The value $\phi_{i^*}$ that minimises \cref{eq: pre-bound consistency} will be of the form $\phi_{i^*} = \lfloor \phi_{i^*} \rfloor + x$, where $0\leq x \leq 1$ and $\lfloor \phi_{i^*} \rfloor$ is some integer in $\mathbb{Z}_{d}$. Thus, there will always be a measurement outcome $k \in \mathbb{Z}_{d}$, such that $k = \lfloor \phi_{i^*}\rfloor$. We now choose some $0 \leq b < \Delta$ and truncate the sum in \cref{eq: pre-bound consistency} in the following manner to obtain a lower bound as
\begin{align}
      \E [P_{X,\Delta}]&\geq  \sum_{|\delta'|\leq \Delta-b} \sinc^{2}(\phi_{i^*} - \lfloor \phi_{i^*}\rfloor - \delta') \notag \\
      &\times\sum_{|\delta|\leq \Delta}\sinc^{2}(\phi_{i^*} - \lfloor \phi_{i^*}\rfloor - \delta' - \delta) \notag   \\
     &= \sum_{|\delta'|\leq \Delta-b} \sinc^{2}(x - \delta')\sum_{|\delta|\leq \Delta}\sinc^{2}(x - \delta' - \delta). \label{eq: Pre-bound semi-final} 
\end{align}
Let 
\begin{subequations}
\begin{align}
   S_{\Delta,b}(x) &\equiv \sum_{|\delta'|\leq \Delta-b} \sinc^{2}(x - \delta') \\
   S_{\Delta,b}(x, \delta')&\equiv \sum_{|\delta|\leq \Delta}\sinc^{2}(x - \delta' - \delta).
\end{align}
\end{subequations}
Thus \cref{eq: Pre-bound semi-final} becomes
\begin{align}
    \E[P_{X,\Delta}] &\geq \sum_{|\delta'|\leq \Delta-b} \sinc^{2}(x - \delta') S_{\Delta,b}(x, \delta').
\end{align}
From \cref{Lemma: Sinc lower bound} we have $\forall x,\delta'$
\begin{equation}
    S_{\Delta,b}(x,\delta') \geq 1 - \frac{2}{\pi^{2}\left(b + x + \frac{1}{2}\right)}.
\end{equation}
We then have
\begin{align}
    \E[P_{X,\Delta}] &\geq \left(1 - \frac{2}{\pi^{2}\left(b + x + \frac{1}{2}\right)}\right)\cdot S_{\Delta,b}(x) \notag \\
    &\geq \min_{x} \left(1 - \frac{2}{\pi^{2}\left(b + x + \frac{1}{2}\right)}\right) \cdot \left(\min_{x} S_{\Delta,b}(x) \right) \notag \\
    &= \left(1 - \frac{2}{\pi^{2}\left(b + \frac{1}{2}\right)}\right) \cdot \left(\min_{x} S_{\Delta,b}(x) \right).
\end{align}
Again, from \cref{Lemma: Sinc lower bound} we have $\forall x$
\begin{equation}
     S_{\Delta,b}(x) \geq 1 - \frac{2}{\pi^{2}(|\Delta - \frac{1}{2}|)}.
\end{equation}
Therefore we get
\begin{align}
    \E[{P_{X,\Delta}}] &\geq \left(1 - \frac{2}{\pi^{2}\left(b + \frac{1}{2}\right)}\right) \cdot \left(1 - \frac{2}{\pi^{2}(|\Delta - \frac{1}{2}|)}\right). \label{eq: Semi-Final latest}
\end{align}
Any value of $0\leq b < \Delta$ in \cref{eq: Semi-Final latest} will give us a valid lower bound. For $\Delta \geq 2$ we can then choose $b = \sqrt{\Delta}$ to obtain the final bound as
\begin{equation}
    \E[{P_{X,\Delta}}] \geq f(\Delta), \label{eq: Final Bound}
\end{equation}
where 

\begin{equation}
    f(\Delta) \equiv \left(1 - \frac{2}{\pi^{2}\left(\sqrt{\Delta} + \frac{1}{2}\right)}\right) \cdot \left(1 - \frac{2}{\pi^{2}(|\Delta - \frac{1}{2}|)}\right). \label{eq: f(Delta)}
\end{equation}
Let us now define another random variable $Q_{X,\Delta}$ as
\begin{equation}
    Q_{X,\Delta} \equiv 1 - P_{X,\Delta}.
\end{equation}
Therefore,
\begin{align}
    \E \left[Q_{X,\Delta}\right] = 1 - \E \left[P_{X,\Delta}\right] \notag \\
    \leq 1 - f(\Delta).
\end{align}
By Markov's inequality \cite{PLESNIAK1990106}, for some $a \geq 0$
\begin{align}
    \Pr[Q_{X,\Delta} \geq a] &\leq \frac{\E[Q_{X,\Delta}]}{a} \notag \\
    &\leq \frac{1 - f(\Delta)}{a}.
\end{align}
Therefore,
\begin{align}
    \Pr[Q_{X,\Delta} < a] &\geq 1 - \frac{1 - f(\Delta)}{a} \notag \\
    \implies \Pr[P_{X,\Delta} > 1- a] & \geq 1 - \frac{1 - f(\Delta)}{a}.
\end{align}
Choosing, $a = \sqrt{1 - f(\Delta)}$ we have,
\begin{align}
    \Pr[P_{X,\Delta} > 1-  \sqrt{1 - f(\Delta)}]  &\geq 1 - \sqrt{1 - f(\Delta)} \notag \\
    \Pr[P_{X,\Delta} > g(\Delta)] & \geq g(\Delta),
\end{align}
where 
\begin{equation}
    g(\Delta) \equiv 1 - \sqrt{1 - f(\Delta)}.
\end{equation}
Let $V_{m_{n-1}}^{n}$ denote the event of successful $n^{th}$ verification for some arbitrary measurement outcome $m_{n-1}$ at the $(n-1)^{th}$ round.
Therefore we have
\begin{align}
    &\Pr[V_{m_{n-1}}^{n}] \notag = P_{m_{n-1},\Delta} \notag \\
    &= \Pr[P_{m_{n-1},\Delta} > g(\Delta)]\cdot\Pr[V_{m_{n-1}}^{n}|P_{m_{n-1},\Delta} > g(\Delta)] \notag \\
    &+ \Pr[P_{m_{n-1},\Delta} \leq g(\Delta)]\cdot\Pr[V_{m_{n-1}}^{n}|P_{m_{n-1},\Delta} \leq g(\Delta)] \notag \\
    &> \Pr[P_{m_{n-1},\Delta} > g(\Delta)]\cdot\Pr[V_{m_{n-1}}^{n}|P_{m_{n-1},\Delta} > g(\Delta)] \notag \\ 
    &= \left(g(\Delta)\right)^{2} \notag \\
    &= \left(1 - \sqrt{1-f(\Delta)}\right)^{2}.
\end{align}
This completes the proof.
\end{proof}
\subsection{Proof of \cref{Theorem: PE-QPUF Existential Unforgeability}}
\label{Proof: Theorem 4}

\begin{proof}(\cref{Theorem: PE-QPUF Existential Unforgeability})
Given an arbitrary adversary guess state 
\begin{equation}
   \ket{\psi_{\mathcal{A}}} = \sum_{i} c_{i} \ket{\phi_{i}},
\end{equation}
we have, for any $m \notin Q_{\mathcal{A}}$,
\begin{align}
    &P_{m,\Delta} \notag \\ &= \bra{\psi_{\mathcal{A}}}M_{m,\Delta}\ket{\psi_{\mathcal{A}}} \notag \\
    &= \sum_{i \in \mathbb{Z}_{D}} |c_{i}|^{2} \sum_{|\delta| \leq \Delta} \sins_{d}^{2}(\phi_{i} - m - \delta) \notag \\
    &\leq \bra{\psi_{\mathcal{A}}}\Pi_{m,\Delta + c}\ket{\psi_{\mathcal{A}}} \notag \\ 
    &+ \max_{|\phi_{i}|> 2\Delta} \sum_{|\delta| \leq \Delta} \sins_{d}^{2}(\phi_{i} - m - \delta)\cdot  \bra{\psi_{\mathcal{A}}}\left(\mathbb{I} - \Pi_{m,\Delta + c}\right)\ket{\psi_{\mathcal{A}}}  \notag \\
    &= \left(1 - \frac{2}{\pi^{2}(c - 1)}  \right)\cdot \bra{\psi_{\mathcal{A}}}\Pi_{m,\Delta + c}\ket{\psi_{\mathcal{A}}} + \frac{2}{\pi^{2}(c - 1)} \label{eq: Max overlap},
\end{align}
where in the last step we have applied \cref{Lemma: Sinc Upper bound} with $c >1$.

We have $\forall k \in Q_{\mathcal{A}}$; $|m - k| \geq 2\Delta$. Consider any $(k',\ket{\psi_{k'}}) \in Q_{\mathcal{A}}$. Combining \cref{eq: Max overlap} with \cref{Authentic Verifiation Probability} we have
\begin{align}
    &\bra{\psi_{k'}}\Pi_{m,\Delta + c}\ket{\psi_{k'}} \notag \\ 
    &\leq 1 - \frac{\left(1 - \sqrt{1 - f(\Delta,1)}\right)^{2} -\frac{2}{\pi^{2}(c-1)} }{1 - \frac{2}{\pi^{2}(c-1)}} \label{eq: Max overlap 2}
\end{align}
Let us choose $c = \frac{\Delta}{2} = 2^{\lambda - 1}$. Thus, \cref{eq: Max overlap 2} becomes
\begin{align}
    \bra{\psi_{k'}}\Pi_{m,\Delta + \frac{\Delta}{2}}\ket{\psi_{k'}} &\leq 1 - (1 - \negl(\lambda)) \\
    &= \negl(\lambda) \notag.
\end{align}
This implies, \begin{align}
    \textrm{span} \{Q_{\mathcal{A}}\} \perp \textrm{span} \{\Pi_{m,\Delta + \frac{\Delta}{2}}\}.
\end{align}
Therefore, we apply \cref{Thm: Existential Unforgeability} with the above mentioned choice of $c$ to obtain 

\begin{align}
    &\E[P_{m,\Delta}] \notag \\
    &\leq \left(1 - \frac{2}{\pi^{2}(\frac{\Delta}{2} - 1)}  \right)\cdot \E [\bra{\psi_{\mathcal{A}}}\Pi_{m,\Delta + \frac{\Delta}{2}}\ket{\psi_{\mathcal{A}}}] + \frac{2}{\pi^{2}(\frac{\Delta}{2} - 1)} \notag \\
    &\leq \left(1 - \frac{2}{\pi^{2}(\frac{\Delta}{2} - 1)}\right)\cdot \left(\frac{1}{\frac{d}{2(\Delta+\frac{\Delta}{2})} - |Q_{\mathcal{A}}|}\right) + \frac{2}{\pi^{2}(\frac{\Delta}{2} - 1)} \notag \\
    &= \negl(\lambda),
\end{align}
where in the second last step we have applied \cref{Thm: Existential Unforgeability}.

This completes the proof.
\end{proof}

\section*{Code Availability}

The codes for the simulations can be found in the following github repository: \url{https://github.com/terrordayvg/Puf_sim}.

\section*{Acknowledgments}\small 
The authors gratefully acknowledge support from the German Federal Ministry of Education and Research (BMBF) through various national initiatives: 

\begin{itemize}
    \item \textbf{6G Communication Systems} via the research hub \textit{6G-life} (Grants 16KISK002 and 16KISK263).
    \item \textbf{Post Shannon Communication (NewCom)} (Grants 16KIS1003K and 16KIS1005).
    \item \textbf{QuaPhySI – Quantum Physical Layer Service Integration} (Grants 16KISQ1598K and 16KIS2234).
    \item \textbf{QTOK – Quantum Tokens for Secure Authentication in Theory and Practice} (Grants 16KISQ038 and 16KISQ039).
    \item \textbf{QUIET – Quantum Internet of Things} (Grants 16KISQ093 and 16KISQ0170).
    \item \textbf{Q.TREX – Resilience for the Quantum Internet} (Grants 16KISR026 and 16KISR038).
    \item \textbf{QDCamnetz – Quantum Wireless Campus Network} (Grants 16KISQ077 and 16KISQ169).
    \item \textbf{QR.X – Quantum link extension} (Grants 16KISQ020 and 16KISQ028).
\end{itemize}

Additionally, the authors acknowledge support from the \textit{6GQT} initiative, funded by the Bavarian State Ministry of Economic Affairs, Regional Development, and Energy.

H.B. also received funding from the German Research Foundation (DFG) within Germany’s Excellence Strategy (EXC-2092 – 390781972), while C.D. received further funding under Grant DE1915/2-1.



\bibliography{bibliography}
\end{document}